\documentclass[noinfoline]{imsart}

\RequirePackage[OT1]{fontenc}
\RequirePackage{amsthm,amsmath}
\RequirePackage{natbib}

\usepackage[hyperindex,breaklinks]{hyperref}

\usepackage{verbatim}
\usepackage{comment}

\allowdisplaybreaks

\usepackage{fullpage}

\newcommand{\E}{\mathbb{E}}
\def\Ez#1{\mathbb{E} \left[ #1 \right]}

\newcommand{\HC}{\mathcal{H}^\mathbb{C}}
\newcommand{\HR}{\mathcal{H}}
\newcommand{\LtwoC}{L^2([0,1],\mathbb{C})}
\newcommand{\LtwoR}{L^2([0,1],\mathbb{R})}
\DeclareMathOperator{\I}{i}
\DeclareMathOperator{\Id}{I}

\DeclareMathOperator{\rank}{rank}

\DeclareMathOperator{\tr}{tr}

\newcommand{\D}[1]{\ensuremath{\operatorname{d}\!{#1}}}

\def\FARMA{FARMA$(p,q)${}}
\def\FARFIMA{FARFIMA$(p,d,q)${}}

\newcommand\Item[1][]{%
  \ifx\relax#1\relax  \item \else \item[#1] \fi
  \abovedisplayskip=0pt\abovedisplayshortskip=0pt~\vspace*{-\baselineskip}}

\usepackage{subfigure}
\usepackage{layout}

\usepackage{dsfont}
\usepackage[mathscr]{eucal}
\usepackage[toc,page]{appendix}
\usepackage{mathrsfs}
\usepackage{color}
\usepackage{pifont}
\usepackage{bm}
\usepackage{latexsym}
\usepackage{amsfonts}
\usepackage{amssymb}
\usepackage{epsfig}
\usepackage{graphicx}
\usepackage{multirow}
\usepackage{enumitem}
\usepackage{mathbbol}
\usepackage[dvipsnames]{xcolor}
\usepackage{placeins}
\usepackage{stackrel}

\newtheorem{theorem}{Theorem}
\newtheorem{proposition}{Proposition}

\usepackage{array}
\newcolumntype{P}[1]{>{\centering\arraybackslash}p{#1}}

\startlocaldefs
\numberwithin{equation}{section}
\theoremstyle{plain}
\endlocaldefs

\begin{document}

\begin{frontmatter}


\title{{\large Spectral Simulation of Functional Time Series}}


\runtitle{Spectral Simulation Functional Time Series}

\begin{aug}
\author{\fnms{Tom{\'a}{\v s}} \snm{Rub{\'i}n}\ead[label=e1]{tomas.rubin@epfl.ch}} \and
\author{\fnms{Victor M.} \snm{Panaretos}\ead[label=e2]{victor.panaretos@epfl.ch}}


\runauthor{T. Rub{\'i}n \& V.M. Panaretos}

\affiliation{Ecole Polytechnique F\'ed\'erale de Lausanne}

\address{Institut de Math\'ematiques\\
Ecole Polytechnique F\'ed\'erale de Lausanne\\
\printead{e1}, \printead*{e2}}

\end{aug}

\begin{abstract}
We develop methodology allowing to simulate a stationary functional time series defined by means of its spectral density operators. Our framework is general, in that it encompasses any such stationary functional time series, whether linear or not. The methodology manifests particularly significant computational gains if the spectral density operators are specified by means of their eigendecomposition or as a filtering of white noise. In the special case of linear processes, we determine the analytical expressions for the spectral density operators of functional autoregressive (fractionally integrated) moving average processes, and leverage these as part of our spectral approach, leading to substantial improvements over time-domain simulation methods in some cases. The methods are implemented as an \texttt{R} package (\texttt{specsimfts}) accompanied by several demo files that are easy to modify and can be easily used by researchers aiming to probe the finite-sample performance of their functional time series methodology by means of simulation.
\end{abstract}


\end{frontmatter}

\tableofcontents

\section{Introduction} \label{intro}

\textit{Functional data analysis} \citep{ramsay2013functional,horvath2012inference,ferraty2006nonparametric} considers statistical problems where the data and parameter spaces are comprised of \emph{functions} and \emph{operators}. The probabilistic models for such data/parameters usually involve notions of random elements in infinite dimensional Hilbert spaces and related (linear) operators, and their theoretical analysis involves many challenges deviating from those  typically encountered with multivariate analysis. Namely, the
analysis of infinite dimensional problems requires tools from functional analysis, while many standard inference problem become ill-posed.
A (temporal) sequence of functional random elements is then called a \textit{functional time series} and constitutes a probabilistic framework for scenarios where functions are collected sequentially and subject to dependencies. Examples of such data include daily profiles of meteorological variables \citep{hormann2010weakly,rubin2020sparsely}, traffic data \citep{klepsch2017prediction}, DNA strings dynamics \citet{tavakoli2016detecting}, or intra-day trading data \citep{cerovecki2019functional}.

The development of functional time series is historically started with the generalisation of univariate or multivariate time series models into infinite dimensions, and has evolved with gradual generalisation. Functional autoregressive (FAR) process was defined by \citet{bosq1999autoregressive,mas2007weak}, while prediction for functional moving average process (FMA) studied by \citet{chen2016functional}, and the two concepts were combined into the functional moving average process (FARMA) by \citet{klepsch2017prediction}. More recently, long-range dependence was incorporated into these models by \citet{li2019long} who defined functional autoregressive fractionally integrated moving average processes (FARFIMA). A detailed treatment of the foundations of linear functional process can be found in \citet{bosq2012linear}.

A different line of development in functional time series domain abandoned the linear processes structure, and investigated more general stationary sequences from the point of view of weak dependence. \citet{hormann2010weakly} studied weakly dependent data and studied the estimation of the long-run covariance operator and \citet{horvath2013estimation} established a central limit theorem for weakly dependent functional data. Additional univariate or multivariate methods have been adapted for the functional time series setting that serve for estimation, prediction, or testing problems \citep{aue2017estimating,aue2015prediction,aue2017functional,laurini2014dynamic,hormann2013functional,gorecki2018testing,gao2019high}.

Parallel to the time domain approaches, the statistical analysis of functional time series has been fruitful also in the spectral domain. The foundations for frequency domain methods were established in \citet{panaretos2013fourier}, while \citet{panaretos2013cramer} and \citet{hormann2015dynamic}  introduced dimension reduction techniques based on the harmonic/dynamic principal component analysis. The spectral domain tools have been successfully used to solve other problems, such as functional lagged regression \citep{hormann2015estimation,pham2018methodology,rubin2019functional}, stationarity testing \citet{horvath2014testing}, periodicity detection \citep{hormann2018testing}, two-sample testing \citet{tavakoli2016detecting}, and white noise testing \cite{zhang2016white}, to mention but a few. The spectral analysis of functional time series was generalised by the introduction of the notion of \textit{weak} spectral density operator \citep{tavakoli2014} that allows for the analysis of long-range dependent functional time series. Some spectral domain results for possibly long-range dependent Gaussian processes are established by \citet{ruiz2019spectral}.

Any methodological development in functional time series will be accompanied by a finite sample performance assessment of the novel method, given the complexity of the data involved.  Such simulations require the generation of functional time series with prescribed model dynamics. Despite many new methods being generally applicable to time series (whether linear or not), their assessments is carried out predominately on simulated data coming from FARMA processes, typically functional AR processes, because their simulation is straightforward in the time-domain by applying the autoregressive equation sequentially on white noise (or a moving average of white noise). In order to assess the applicability of a method beyond linear processes, however, one should aim to cover as broad as possible a range of possible functional time series dynamics (including nonlinear dynamics). This is especially true for methods that are not specific to linear processes but whose assumptions, theory, and implementaton are more generally valid. Indeed, many functional time series methods \citep{hormann2015dynamic,hormann2015estimation,zhang2016white,tavakoli2016detecting} rely on the eigendecomposition of spectral density operators (the harmonic/dynamic principal components) and present performance tradeoffs that are best captured by their spectral structure. It is thus beneficial to be able to simulate functional time series specified by means of their spectral density structure.

The objective of this article is to develop a general-purpose simulation method that is able to efficiently simulate stationary functional time series not restricted to the linear class. The approach is to use the spectral specification of such a time series, by means of its \emph{spectral density operator}. The general method, presented in Section~\ref{sec6:simulation_in_spectral_domain}, hinges on a discretisation and dimension reduction of the functional Cram\'er representation \citep{panaretos2013cramer}. It simulates an ensemble of independent complex random elements whose covariance operators match the designated spectral density operators, and transposes this ensemble into the time-domain by the means of the (inverse) fast Fourier transform. We show that this strategy is particularly effective when the series is defined by means of the eigendecomposition of its spectral density operator or by filtering a white noise, but consider various other specification scenarios, too. For FARMA and FARFIMA processes, in particular, we develop analytical expressions for their spectral density operators, and exploit these in conjunction with spectral methods. To our knowledge, the spectral density operators for these processes, while being infinite-dimensional analogues of the univariate/multivariate versions \citep{priestley1981spectral_1,priestley1981spectral_2}, have not yet been previously rigorously established in functional time series literature.

Our functional time series simulation method in the spectral domain is inspired in part by the methods for scalar and multivariate time series simulation. The original idea of simulating a signal in the spectral domain and converting it to the time-domain by the inverse fast Fourier transform seems to be due to \citet{thompson1973generation}. This approach was further explored by \citet{percival1993simulating} who reviewed some variants of the algorithm and addressed some practical implementation questions, and \citet{davies1987tests} used the method for simulation of fractionally integrated noise processes.
Furthermore, the simulation of multivariate time series with given spectral density matrices is due to \citet{chambres1995simulation}.
However, pushing the general ideas forward to functional time series is not a matter of simple generalisation of the multivariate time series simulation methods. The intrinsic infinite dimensionality of functional data calls for the approximate generation of infinite dimensional objects approximated in finite dimension, which requires optimally reducing dimension (which we implement either via the Karhunen-Lo\'eve or the Cram\'er-Karhunen-Lo\`eve representation \citep{panaretos2013cramer})  and/or judicious discretisation (pixelisation) of the spatial domain (the argument of each function). An additional side effect of this, in contrast to the multivariate case, is that one must pay particular attention that the simulation algorithms scale well as the discretisation resolution refines and the dimension parameter grows, and these need to be incorporated in the time complexity assessments.

Our spectral domain simulation method constitutes a general approach, able to simulate arbitrary functional time series that are specified in the frequency domain, with additional computational speed-ups that can be realised when assuming a special structure of the spectral density operators. In particular, simulation of the important \FARFIMA{} processes can be much faster in the spectral domain than in the time-domain, while the spectral domain simulation of \FARMA{} processes is competitive with time-domain methods.

The rest of the article is structured as follows:
Section~\ref{sec:framework} introduces the functional time series framework with special attention to their (doubly) spectral analysis and includes the aforementioned novel derivation of the spectral density operators of FARMA and FARFIMA processes as Theorem~\ref{thm:FARMA-part-ii} and \ref{thm:FARFIMA-part-ii} respectively. Section~\ref{sec6:simulation_in_spectral_domain} presents the high-level spectral domain simulation algorithm along with a discussion of its various implementation as subsections.
Section~\ref{sec6:examples} provides with concrete examples followed by a short benchmark simulation study. Section~\ref{sec:general_redommendations} concludes the article by summarising key features and qualities of the proposed simulation methods, along with some recommendations for practitioners.

The article is accompanied by an \texttt{R} package \texttt{specsimfts} (Section~\ref{sec:code_availability}) that implements all the proposed methods and includes several demo files that are easy to modify and can be easily made use of by practitioners.

\section{Functional Time Series Framework}
\label{sec:framework}

\subsection{Spectral Analysis of Functional Time Series}

We will throughout work in a real separable Hilbert space denoted as $\HR$ with inner product $\langle f,g \rangle,\,f,g\in\HR$ and induced norm $\|f\|,\,f\in\HR$. The complexification of $\HR$ is denoted as $\HC$ and we maintain the same notation for the inner product $\langle\cdot,\cdot\rangle$ and norm $\|\cdot\|$ on $\HC$. Though parts of the functional time series theory presented in this section are valid for any such $\HR$ and $\HC$, the simulation methods are tailored to the space of real square-integrable functions defined on $[0,1]$, denoted as $\LtwoR$. The inner product on $\LtwoR$, or its complexification $\LtwoC$, is defined as $\langle g_1,g_2 \rangle = \int_0^1 f(x)\overline{g(x)}\D x,\,f,g\in\HR$ (or $\in\HC$), and the norm $\|f\| = (\int_0^1 |f(x)|^2\D x)^{1/2},\,f\in\HR$  (or $\in\HC$).
The space of the bounded linear operators acting on $\HR$ and $\HC$ is denoted $\mathcal{L}(\HR)$ and $\mathcal{L}(\HC)$ respectively and the corresponding operator norm  as $\|\cdot\|_{\mathcal{L}(\HR)}$ and $\|\cdot\|_{\mathcal{L}(\HC)}$ respectively.

The classical approach in functional data analysis is to probabilistically model the functional data as random elements in the Hilbert space $\HR$. Considering $Z$ to be a random element in $\HR$ with a finite second moment $\E \|Z\|^2<\infty$, we define its \textit{mean function} as
$ \mu_Z = \E Z\in \HR $ and the \textit{covariance operator}
$$ \mathscr{R}^Z = \Ez{ (Z-\mu) \otimes (Z-\mu) } = \Ez{ \langle \cdot, Z-\mu\rangle (Z-\mu) },$$
where $x \otimes y$ denotes the tensor product of $x,y\in\HC$ defined as the operator $x \otimes y : \HC\to\HC,\, v \mapsto \langle v, y \rangle x$. The covariance operator $\mathscr{R}^Z$ is a self-adjoint positive-definite trace class operator.

A \textit{(real) functional time series} is conceptualized as a time ordered sequence of random elements in $\HR$ and is denoted as $X \equiv \{X_t\}_{t\in\mathbb{Z}}$.
Throughout this article we work with functional time series with finite second moments, i.e. $\E \|X_t\|^2 < \infty,\,t\in\mathbb{Z}$, and which are second-order stationary in the time variable $t$.
If we additionally assume the random curves perspective, i.e. assuming $\HR$ to be the function space $\LtwoR$, it is common 
 to assume that the individual sample paths (trajectories) of the random curves are continuous.
In this case, a functional time series can be interpreted pointwise as a sequence of random curves $X\equiv \{X_t(x):x\in[0,1]\}_{t\in\mathbb{Z}}$. The index variable $t$ is interpreted as a discrete time parameter, and argument variable $x$ can often be interpreted as a continuous spatial location in the domain $[0,1]$, and we choose to refer to $x$ as the spatial location for clarity.

Under the above stated assumptions we may define the first and second order characteristics of the functional time series $X\equiv \{X_t\}_{t\in\mathbb{Z}}$, namely the \textit{mean function}
$\mu_X = \E X_0$
and, for $h\in\mathbb{Z}$, the \textit{lag-$h$ autocovariance operator}
$$ \mathscr{R}^X_h = \Ez{ \left( X_{h} - \mu_X \right) \otimes \left( X_0 - \mu_X \right) } 
= \Ez{ \left\langle \cdot, X_0 - \mu_X \right\rangle \left( X_{h} - \mu_X \right) } . $$
To simplify the notation and the presentation we shall only consider the centred functional time-series, i.e. $\mu\equiv 0$, in order to focus on second order structure, which is the essential part for simulation purposes.


We now review key aspects of the analysis of functional time series in the spectral domain. First, we consider functional time series satisfying  \textit{weak dependence} conditions, manifested in one of the following norms:
\begin{align}
\label{eq6:weak_dependence_trace_norm}
\sum_{h\in\mathbb{Z}} \left\| \mathscr{R}^X_h \right\|_1 &< \infty, \\
\label{eq6:weak_dependence_HS_norm}
\sum_{h\in\mathbb{Z}} \left\| \mathscr{R}^X_h \right\|_2 &< \infty, \\
\label{eq6:weak_dependence_op_norm}
\sum_{h\in\mathbb{Z}} \left\| \mathscr{R}^X_h \right\|_{\mathcal{L}(\HR)} &< \infty
\end{align}
where $\|\cdot\|_1$, $\|\cdot\|_2$, $\|\cdot\|_{\mathcal{L}(\HR)}$ denote the trace-class norm, the Hilbert-Schmidt norm, and the operator norm respectively.
The \emph{spectral density operator} was first defined under \eqref{eq6:weak_dependence_trace_norm} by \citet{panaretos2013fourier}, under the slightly weaker assumption \eqref{eq6:weak_dependence_HS_norm} by \citet{hormann2015dynamic}, and finally under \eqref{eq6:weak_dependence_op_norm} by \citet{tavakoli2014}. Because \eqref{eq6:weak_dependence_op_norm} is the weakest condition of the three, we shall be working with this assumption, under which the {spectral density operator} is defined by the formula \citep{tavakoli2014}[Proposition 2.3.5]
\begin{equation}\label{eq6:definition_spectral_density_operator}
\mathscr{F}^X_\omega = \frac{1}{2\pi} \sum_{h\in\mathbb{Z}} \mathscr{R}^X_h e^{-\I h \omega}
\end{equation}
where the sum converges in $\|\cdot\|_{\mathcal{L}(\HC)}$ at each $\omega\in[0,2\pi]$.
The spectral density operator $\mathscr{F}^X_\omega$ is self-adjoint, non-negative definite and trace-class for each $\omega\in[0,2\pi]$
and the inversion formula holds in $\|\cdot\|_{\mathcal{L}(\mathcal{H})}$:
\begin{equation}\label{eq6:spectral_density_operator_inverse_formula}
\mathscr{R}^X_h = \int_0^{2\pi} \mathscr{F}^X_\omega e^{\I h\omega} \D \omega, \qquad h\in\mathbb{Z}.
\end{equation}

Furthermore, whenever
\begin{equation}\label{eq6:weak_dependence_traces}
\sum_{h\in\mathbb{Z}} \left|\tr(\mathscr{R}^X_h)\right| < \infty,
\end{equation}
the spectral density operator is uniformly bounded
$$ \sup_{\omega\in[0,2\pi]} \left\| \mathscr{F}^X_\omega \right\|_1 \leq \frac{1}{2\pi} \sum_{h\in\mathbb{Z}} \left|\tr(\mathscr{R}^X_h)\right| < \infty $$
and
$$ \sup_{h\in\mathbb{Z}} \left\| \mathscr{R}^X_h \right\|_1 \leq \sum_{h\in\mathbb{Z}} \left|\tr(\mathscr{R}^X_h)\right| < \infty .$$

Finally, the definition of spectral density operator can be relaxed into the notion of the \textit{weak} spectral density operator \citep{tavakoli2014}. Denote $\mathcal{L}_1(\HC)$ the space of trace-class operators on $\HC$. If there exists a function $\mathscr{F}^X : [0,2\pi] \to \mathcal{L}_1(\HC)$ defined almost everywhere on $[0,2\pi]$ such that $\int_0^{2\pi}\| \mathscr{F}^X_\omega \|_1 \D\omega < \infty$
and the inversion formula\eqref{eq6:spectral_density_operator_inverse_formula} holds,
then $\mathscr{F}^X$ is called the \textit{weak spectral density operator} of $X$.
If the weak spectral density operator exists it is defined uniquely only almost everyone on $[0,2\pi]$. This is a consequence of the fact that $\mathscr{F}^X$ is defined as an element of the Bochner space $L^1( [0,2\pi], \mathcal{L}_1(\HC) )$.
That being said, under the weak dependence \eqref{eq6:weak_dependence_op_norm}, the spectral density operator \eqref{eq6:definition_spectral_density_operator} is also the weak spectral density operator.

Though the definition of the weak spectral density operator appears rather abstract, it is in fact required for the spectral analysis of  long-range dependent FARFIMA processes (considered in Section~\ref{subsec:FARFIMA}) which do not satisfy the assumption \eqref{eq6:weak_dependence_op_norm} but will be shown to admit a weak spectral density operator.

%

Lastly we point out that we opt for presenting the spectral theory with the spectral domain $[0,2\pi]$, as opposed to $[-\pi,\pi]$ often adopted in literature \citep{panaretos2013cramer,tavakoli2014,hormann2015dynamic}, because its connections to the simulation methods based on discrete (fast) Fourier transform in Section~\ref{sec6:simulation_in_spectral_domain} are more transparent. These two perspectives are equivalent and can be easily interchanged by the $2\pi$-periodicity
$$
\mathscr{F}^X_{-\omega} = \mathscr{F}^X_{2\pi-\omega}, \qquad \omega\in[0,\pi].
$$

\subsection{The Cram\'{e}r-Karhunen-Lo\`{e}ve Representation}
\label{subsec:CKL}

The classical Karhunen-Lo\`{e}ve expansion decomposes i.i.d. functional data into uncorrelated components and achieves optimal dimensionality reduction at the same time. It has consequently been used as a main tool for simulating independent functional data. The situation for functional time series data becomes more involved due to the dependence between curves, and using a similar decomposition for the purpose of simulation will now require two steps.
Firstly, the Cram\'er representation (Proposition~\ref{prop:cramer} and \eqref{eq6:cramer}), which separates the functional time series into distinct uncorrelated frequencies. And, in addition to that, applying the ideas of the classical Karhunen-Lo\`eve expansion at each frequency to obtain the Cram\'{e}r-Karhunen-Lo\`eve representation (Proposition~\ref{prop:optimality_CKL} and \eqref{eq6:CKL_truncated}). We now review these two representations because they, together with their discretised approximations \eqref{eq6:cramer_approx} and \eqref{eq6:CKL_approx_truncated}, will provide the basis for our simulation method presented in Section~\ref{sec6:simulation_in_spectral_domain}.

Before venturing into the spectral domain, we recall the classical Karhunen-Lo\`{e}ve expansion \citep{karhunen1946spektraltheorie,loeve1946fonctions,ash2014topics,grenander1981abstract}. Let $\{X_t\}$ be i.i.d. zero-mean square-integrable random elements in $\HR$ and denote the eigendecomposition of the corresponding covariance operator as
$ \mathscr{R}^X_0 = \sum_{n=1}^\infty \lambda_n \varphi_n \otimes \varphi_n $
where $\{\lambda_n\}_{n=1}^\infty$ are the eigenvalues of $\mathscr{R}^X_0$ and $\{\varphi_n\}_{n=1}^\infty$ their associated eigenfunctions. Then, the classical Karhunen-Lo\`{e}ve expansion relies on truncating the sum
$$ X_t = \sum_{n=1}^\infty \sqrt{\lambda_n} \xi^{(t)}_n \varphi_n $$
where $\xi^{(t)}_n = \langle X_t, \varphi_n \rangle / \sqrt{\lambda_n}$.
The mode of convergence depends on the regularity of $\mathscr{R}^X_0$, but convergence in expected squared Hilbert norm is always valid when $\mathscr{R}^X_0$ is trace-class.

In order to take into account the temporal dependence one begins by decomposing the time series into distinct frequencies, a step made rigorous by means of the functional Cram\'{e}r representation, due to \citet{panaretos2013cramer}[Theorem 2.1] and  \citet{tavakoli2014}[Theorem 2.4.3]. We combine the two statements into a single statement, to be used for our purposes, below:

\begin{proposition}[Functional Cram\'{e}r representation]
\label{prop:cramer}
Let the functional time series $X\equiv\{X_t\}_{t\in\mathbb{Z}}$ admit the weak spectral density operator $\mathscr{F}^X \in L^p([0,2\pi], \mathcal{L}_1(\HC)$ for some $p\in(1,\infty]$. Then $X$ permits the functional Cram\'{e}r representation
\begin{equation}\label{eq6:cramer}
X_t = \int_0^{2\pi} e^{\I t\omega} \D Z_\omega, \qquad\text{almost surely}.
\end{equation}
where stochastic integral \eqref{eq6:cramer} can be understood in Riemann–Stieltjes limit sense
\begin{equation}\label{eq6:cramer_riemann_stieltjes}
\Ez{ \left\| X_t - \sum_{k=1}^K e^{\I t \omega_k} \left( Z_{\omega_{k+1}} - Z_{\omega_k} \right) \right\|^2 } \to \infty,
\qquad\text{as}\quad K\to\infty,
\end{equation}
where $0=\omega_1 < \dots < \omega_{k+1} = 2\pi$ and $\max |\omega_{k+1}-\omega_k| \to 0$ as $K\to\infty$.
For each $\omega\in[0,2\pi]$, $Z_\omega$ is a random element in $\HC$ defined by
\begin{equation}\label{eq6:cramer_definition_Z}
Z_\omega = \lim_{T\to\infty} \sum_{|t|<T} \left( 1 + \frac{|t|}{T} \right) g_\omega(t) X_{-t}
\end{equation}
where the limit holds with respect to $\E\|\cdot\|^2$ and
$$ g_\omega(t) = \frac{1}{2\pi} \int_{0}^{\omega} e^{-\I t\alpha} \D\alpha, \qquad \omega\in[0,2\pi]. $$
Moreover, the process $\{Z_\omega\}_{\omega\in[0,2\pi]}$ satisfies $\E[ \| Z_\omega\|^2_2 ] = \int_0^\omega \|\mathscr{F}^X_\alpha\|_1 \D\alpha$, $ \E[ Z_\omega \otimes Z_{\omega'} ] = \int_0^{\min(\omega,\omega')} \mathscr{F}^X_\alpha \D\alpha $ for $\omega,\omega'\in [0,2\pi]$ and has orthogonal increments
$$ \E\left\langle Z_{\omega_1}-Z_{\omega_2}, Z_{\omega_3} - Z_{\omega_4} \right\rangle = 0$$
with $\omega_1 > \omega_2 \geq \omega_3 > \omega_4.$
\end{proposition}

The Cram\'{e}r representation \eqref{eq6:cramer} provides a scheme for decomposing $X$ into distinct frequencies. For $0=\omega_1 < \dots < \omega_{k+1} = 2\pi$ we have an approximation by \eqref{eq6:cramer_riemann_stieltjes}
\begin{equation}\label{eq6:cramer_approx}
X_t \approx \sum_{k=1}^K e^{\I t\omega_k} \left( Z_{\omega_{k+1}} - Z_{\omega_k} \right).
\end{equation}
The approximation \eqref{eq6:cramer_approx} essentially decomposes the functional time series $\{X_t\}_{t\in\mathbb{Z}}$ into uncorrelated components $Z_{\omega_{k+1}} - Z_{\omega_k},\,k=1,\dots,K$.
Heuristically, the covariance operator of the increment $Z_{\omega_{k+1}} - Z_{\omega_k}$ is expected to be close to $\mathscr{F}^X_{\omega_k} (\omega_{k+1}-\omega_k)$.
By virtue of being a non-negative definite operator, the spectral density operator $\mathscr{F}^X_{\omega}$, admits a spectral decomposition of its own at each frequency $\omega$,
\begin{equation}\label{eq6:spectral_density_operator_harmonic_decomposition}
\mathscr{F}^X_\omega = \sum_{n=1}^\infty \lambda_n(\omega) \varphi_n(\omega)\otimes\varphi_n(\omega)
\end{equation}
where $ \{\lambda_n(\omega)\}_{n=1}^\infty $ are the eigenvalues of $\mathscr{F}^X_{\omega}$, called the \textit{harmonic eigenvalues}, and their associate eigenfunctions $\{ \varphi_n(\omega) \}_{n=1}^\infty$, called the \textit{harmonic eigenfunctions}.
 This suggests a second level of approximation, namely using the Karhunen-Lo\`eve expansion to write
$$
X_t \approx \sum_{k=1}^K e^{\I t\omega_k} \sum_{n=1}^\infty \xi_n^{(k)} \varphi_n(\omega_k)
$$
with $\xi_n^{(k)} = \langle  Z_{\omega_{k+1}} - Z_{\omega_k}, \varphi_n(\omega_k) \rangle / \sqrt{\lambda_n(\omega_k)}$ and then truncating at $N\in\mathbb{N}$
\begin{equation}\label{eq6:CKL_approx_truncated}
X_t \approx \sum_{k=1}^K e^{\I t\omega_k} \sum_{n=1}^N \xi_n^{(k)} \varphi_n(\omega_k) .
\end{equation}

The approximation \eqref{eq6:CKL_approx_truncated} consists of finite number of uncorrelated random variables $\xi_n^{(k)},\,k=1\dots,K,\,n=1,\dots,N$  and will serve as the basis for our simulation method described in Section \ref{subsec6:simulation_CKL}.
To rigorously define this approach, and show its optimality, we must consider the stochastic integral
\begin{equation}\label{eq6:stoch_integral}
\int_0^{2\pi} e^{\I t\omega} C(\omega)\D Z_\omega
\end{equation}
which can be defined by the means similar to the It\^{o} stochastic integral, rigorously proved in \citet{panaretos2013cramer} and \citet{tavakoli2014}.
If $\mathscr{F}^X \in L([0,2\pi], \mathcal{L}_1(\HC) ) $ for $p\in(1,\infty]$, then \eqref{eq6:stoch_integral} is well defined for $C\in\mathbb{M}$ where $\mathbb{M}$ is the completion of $L^{2q}( [0,2\pi], \mathcal{L}(\HC) )$ with respect to the norm $\|\cdot\|_\mathbb{M} = \sqrt{\langle \cdot,\cdot \rangle_{\mathbb{M}}}$ where
$$ \langle A,B \rangle_{\mathbb{M}} = \int_0^{2\pi} \tr\left( A(\omega) \mathscr{F}^X_\omega B(\omega)^*  \right)\D\omega,
\qquad A,B\in\mathbb{M}. $$
In this notation, one has  (\cite{panaretos2013cramer}[Theorem 3.7], \citet{tavakoli2014}[Theorem 2.8.2]):

\begin{proposition}[Optimality of Cram\'{e}r-Karhunen-Lo\`{e}ve representation]
\label{prop:optimality_CKL}
Let the functional time series $X\equiv\{X_t\}_{t\in\mathbb{Z}}$, satisfying the functional Cram\'{e}r representation \eqref{eq6:cramer},  admit the weak spectral density operator $\mathscr{F}^X \in L^1([0,2\pi], \mathcal{L}_1(\HC))$ such that the function $\omega\in[0,2\pi]\mapsto \mathscr{F}^X_\omega$ is continuous on $[0,2\pi]$ with respect to the operator norm $\|\cdot\|_{\mathcal{L}(\HC)}$ and all the non-zero harmonic eigenvalues of $\mathscr{F}^X_\omega$ are distinct, $\omega\in[0,2\pi]$. Let
$$ X_t^* = \int_0^{2\pi} e^{\I t\omega} C(\omega) \D Z_\omega $$
with $C \in \mathbb{M}$.
Let $N: [0,2\pi]\to\mathbb{N}$ be a c\`{a}dl\`{a}g function.
Then, the solution to
\begin{align*}
&\min \Ez{ \left\| X_t - X_t^* \right\|^2} \\
\text{subject to}\quad
& \rank( C(\omega) ) \leq N(\omega)
\end{align*}
is given by
$$ C(\omega) = \sum_{n=1}^{N(\omega)} \varphi_n(\omega) \otimes \varphi_n(\omega).$$
Moreover, the approximation error is given by
$$ \Ez{ \left\| X_t - X_t^* \right\|^2} = \int_0^{2\pi} \left\{ \sum_{n=N(\omega)+1}^\infty \lambda_n(\omega) \right\} \D \omega .$$
\end{proposition}

Proposition~\ref{prop:optimality_CKL} justifies that the process
\begin{equation}\label{eq6:CKL_truncated}
 X^*_t = \int_0^{2\pi}  \sum_{n=1}^N e^{\I t\omega} \left(\varphi_n(\omega) \otimes \varphi_n(\omega)\right) \D Z_\omega 
\end{equation}
yields optimal dimension reduction when we set the rank requirement $N(\omega)\equiv N\in\mathbb{N}$ uniformly across all frequencies. Although the definition of the finite dimensional reduction \eqref{eq6:CKL_truncated} appears quite abstract, it turns out that one can represent $X^*$ in one-to-one manner as an $N$-dimensional multivariate time series using a particular choice of the filter of the original time series $X$. Because our simulation method presented in Subsection~\ref{subsec6:simulation_CKL} is based directly on the approximations \eqref{eq6:CKL_approx_truncated} and \eqref{eq6:CKL_truncated}, we do not pursue the multivariate time series representation here and refer the reader to \citet{panaretos2013cramer,tavakoli2014,hormann2015dynamic}.


\subsection{Spectral Analysis of \FARMA{} Processes}
\label{subsec:FARMA}

Linear models for processes in function spaces have been extensively studied in the literature, and many classical time series models from the scalar or vector time series domain have been gradually generalised to infinite dimensions. Functional autoregressive processes have been treated in depth by \citet{bosq2012linear} and \citet{mas2007weak}, and functional moving average process by \citet{chen2016functional}. Their combination, the
functional autoregressive moving average (FARMA) mocel, has been presented by 
\citet{klepsch2017prediction}. In the following text we recall the time domain analysis of FARMA  processes and then develop our new results on the frequency domain analysis thereof.

The \FARMA{} process, $p,q\in\mathbb{N}_0$, is a sequence $X=\{ X_t \}_{t\in\mathbb{Z}}$ of random $\HR$-elements, satisfying the equation
\begin{equation}\label{eq6:FARMA_def}
X_t = \sum_{j=1}^p \mathcal{A}_j X_{t-j} + \epsilon_t + \sum_{j=1}^q \mathcal{B}_j \epsilon_{t-j}, \qquad t\in\mathbb{Z},
\end{equation}
where $\mathcal{A}_1,\dots,\mathcal{A}_p$ and $\mathcal{B}_1,\dots,\mathcal{B}_q$ are bounded linear operators and $\{\epsilon_t\}_{t\in\mathbb{Z}}$ is a sequence of zero-mean i.i.d. random elements in $\HR$ with the covariance operator $\mathcal{S}$.

The time-domain analysis of the \FARMA{} process was considered by \citet{klepsch2017prediction}, who in particular established:
\begin{theorem}[{\citet{klepsch2017prediction}}]
\label{thm:FARMA-part-i}
Assume that there exists $j_0\in\mathbb{N}$ such that the operator
$$
\tilde{\mathcal{A}} =
\begin{bmatrix}
\mathcal{A}_1 & \cdots & \mathcal{A}_{p-1} & \mathcal{A}_p \\
\Id &   &  & 0 \\
  & \ddots &  & \vdots \\
 &  & \Id & 0 \\
\end{bmatrix}
$$
satisfies
\begin{equation}\label{eq6:FARMA_condition}
\| \tilde{\mathcal{A}}^{j_0} \|_{\mathcal{L}(\HR^p)} < 1
\end{equation}
where $\Id$ is the identity operator on $\HR$ and $\|\cdot \|_{\mathcal{L}(\HR^p)}$ denotes the operator norm on $\mathcal{L}(\HR^p)$, the space of bounded linear operators acting on the product space $\HR^p=\HR\times\cdots\times\HR$.
Then the \FARMA{} process defined by \eqref{eq6:FARMA_def} is uniquely defined, stationary, and causal.
\end{theorem}

We now show that, under the same assumptions as those by \citet{klepsch2017prediction}, we may analyse characterise the \FARMA{} process in the spectral domain:
\begin{theorem}\label{thm:FARMA-part-ii}
Under the assumptions of Theorem~\ref{thm:FARMA-part-i}, the process satisfies the weak dependence condition \eqref{eq6:weak_dependence_op_norm} with $\mathscr{R}^X_h$, and its spectral density operator at frequency $\omega\in[0,2\pi]$ is given by
\begin{equation}\label{eq6:FARMA_spectral_density_operator}
\mathscr{F}^X_\omega = \frac{1}{2\pi}
\mathscr{A}( e^{-\I\omega} )^{-1}
\mathscr{B}( e^{-\I\omega} )
\mathcal{S}
\mathscr{B}( e^{-\I\omega} )^*
\left[\mathscr{A}( e^{-\I\omega} )^* \right]^{-1}
\end{equation}
where
\begin{align}
\label{eq6:FARMA_spectral_density_operator_def_A}
\mathscr{A}(z) &= \Id - \mathcal{A}_1 z - \dots - \mathcal{A}_p z^p, \\
\label{eq6:FARMA_spectral_density_operator_def_B}
\mathscr{B}(z) &= \Id + \mathcal{B}_1 z + \dots + \mathcal{B}_p z^q.
\end{align}
are $\mathcal{H}$-valued polynomials in the variable $z\in\mathbb{C}$.
\end{theorem}
Theorem~\ref{thm:FARMA-part-ii} is proved in Appendix~\ref{subsec:proof_of_thm:FARMA}.

\subsection{Spectral Analysis of \FARFIMA{} Process}
\label{subsec:FARFIMA}

Long range dependence (a.k.a. long memory) is a well known phenomenon in time series analysis, consisting in a time series exhibiting slow decay of its temporal dependence \citep{hurst1951long,mandelbrot1968fractional,beran1994statistics,palma2007long}. The need to model and analyse such series has led to the definition of autoregressive fractionally integrated moving average (ARFIMA) processes
\citep{granger1980introduction,hosking1981fractional}. Such long-range dependencies have also been detected functional time series, for example in series of daily volatility \citep{casas2008econometric}, and inspired the theoretical framework of long-range dependent functional time series model \citep{li2019long} and associated estimation methods \citep{shang2020comparison}.

\citet{li2019long} defined the functional ARFIMA process (FARFIMA) which and we recall its definition, before deriving its spectral analysis that will allow an efficient simulation of its realisations in Section~\ref{sec6:simulation_in_spectral_domain}.

The \FARFIMA{} model with $p,q\in\mathbb{N}_0$ and $d\in(-1/2,1/2)$ models a sequence $\tilde{X}=\{ \tilde{X}_t \}_{t\in\mathbb{Z}}$ of random $\HR$-elements via the equation
\begin{equation}\label{eq6:FARIMA_def}
(\Id-\Delta)^d \tilde{X}_t = X_t
\end{equation}
where $\Delta$ is the backshift operator and $X=\{ X_t \}_{t\in\mathbb{Z}}$ is the \FARMA{} process defined via equation \eqref{eq6:FARMA_def}.
When $d=0$, the \FARFIMA{} reduces to the \FARMA{} model.

\citet{li2019long} established the existence and uniqueness results of the \FARFIMA{} process and its time-domain properties:
\begin{theorem}[{\citet{li2019long}}]
\label{thm:FARFIMA-part-i}
The \FARFIMA{} process $\tilde{X} = \{\tilde{X}_t\}_{t\in\mathbb{Z}}$ with $p,q\in\mathbb{N}_0$ and $d\in(-1/2,1/2)$ defined by the equation \eqref{eq6:FARIMA_def} exists and constitutes a uniquely defined stationary causal functional time series
provided the autoregressive part satisfies the condition \eqref{eq6:FARMA_condition}.
Furthermore, if $d\in(0,1/2)$ the \FARFIMA{} process exhibits the long-memory dependence.
\end{theorem}

Under the same assumptions as \citet{li2019long} we now determine the analytical expression of the spectral density operators of the \FARFIMA{} process:

\begin{theorem}\label{thm:FARFIMA-part-ii}
Under the assumptions of Theorem~\ref{thm:FARFIMA-part-i}, the \FARFIMA{} process admits the weak spectral density $\mathscr{F}^{\tilde{X}} \in L^1( [0,2\pi], \mathcal{L}_1( \HC ) )$ satisfying
\begin{equation}\label{eq6:FARIMA_spectral_density_operator}
\mathscr{F}^{\tilde{X}}_\omega =
\frac{1}{2\pi}
\left[ 2\sin\left(\frac{\omega}{2}\right) \right]^{-2d}
\mathscr{A}( e^{-\I\omega} )^{-1}
\mathscr{B}( e^{-\I\omega} )
\mathcal{S}
\mathscr{B}( e^{-\I\omega} )^*
\left[\mathscr{A}( e^{-\I\omega} )^* \right]^{-1},\quad\omega\in(0,2\pi),
\end{equation}
where $\mathscr{A}$ and $\mathscr{B}$ are given at
\eqref{eq6:FARMA_spectral_density_operator_def_A} and \eqref{eq6:FARMA_spectral_density_operator_def_B}.
The lag-$h$ autocovariance operators of $\tilde{X}$ satisfy
$$ \mathscr{R}_h^{\tilde{X}} = \int_0^{2\pi} \mathscr{F}_\omega^{\tilde{X}} e^{\I h\omega} \D\omega,\qquad h\in\mathbb{Z}. $$
\end{theorem}
Theorem~\ref{thm:FARFIMA-part-ii} is proved in Appendix~\ref{subsec:proof_of_thm:FARFIMA}.

Note that for $d>0$, the term $[ 2\sin(\omega/2) ]^{-2d}$ in formula \eqref{eq6:FARIMA_spectral_density_operator} is unbounded in the neighbourhood of $0$ (and $2\pi$ due to the symmetry). The spectral density being unbounded in the neighbourhood of zero is quintessential also for the univariate ARFIMA processes \citep{hosking1981fractional}.

\section{Simulation of Functional Time Series with Given Spectrum}
\label{sec6:simulation_in_spectral_domain}



In this subsection we will present a functional time series simulation method in the spectral domain. We focus our presentation on functional time series with values in $\LtwoR$ whose trajectories are continuous and whose spectral density operators are integral operators with continuous kernels, but note that our discussion equally applies to other function spaces constituting separable Hilbert spaces.

The objective of the simulation is to generate a Gaussian sample $X_1,\dots,X_T$ for some $T\in\mathbb{N}$ given the spectral density operator $\{ \mathscr{F}_\omega^X \}_{\omega\in[0,2\pi]}$. Without loss of generality, we assume that $T$ is even and we furthermore define the canonical frequencies $\omega_k = (2\pi k)/T,\,k=1,\dots,T$.


At a high level, our spectral domain simulation methods mimics the discrete approximation of the Cram\'er representation \eqref{eq6:cramer_approx}, which boils down to performing the following two steps.

\begin{enumerate}
\item
Generate an ensemble of independent complex mean-zero Gaussian random elements $Z_k',\,k=1,\dots,T/2,T$ such that
\begin{equation}\label{eq6:simulate_Z_covariance_operator}
\Ez{ Z_k' \otimes Z_k' } = \mathscr{F}^X_{\omega_k},\qquad k=1,\dots,T/2,T,
\end{equation}
and, for $k=1,\dots,T/2-1$, generate independent copies $Z_k''$ thereof. 
Define
\begin{equation}\label{eq6:simulation_Z_definition}
Z_k =
\begin{cases}
	\sqrt{2} Z_k'  & k=T/2,T, \\
	Z_k' + \I Z_k''  & k=1,\dots,T/2-1, \\
	Z_{T-k}' - \I Z_{T-k}'' & k=T/2+1,\dots,T/2-1.
   \end{cases}
\end{equation}

\item By the inverse fast Fourier transform algorithm calculate
\begin{equation}
\label{eq6:simulation_iFFT}
X_t = \left(\frac{\pi}{T}\right)^{1/2}
\sum_{k=1}^T Z_k e^{\I t\omega_k},\qquad t=1,\dots,T.
\end{equation}
The formula \eqref{eq6:simulation_Z_definition} ensures that the sample of $\{Z_k\}$ is symmetric and thus inverse Fourier transform constitutes a real-valued functional time series, as will be proved later in Theorem~\ref{theorem6:abstract_method}.
\end{enumerate}

While the application of the inverse fast Fourier transform in Step 2 of the algorithm is computationally fast, the generation of the complex random elements $\{Z_k'\}$ in Step 1, whose covariance operators may in general have no structure in common, is not a trivial matter, and is discussed in the next three subsections, for three different specifications of the operator $\mathscr{F}^X_{\omega_k}$. In Subsection~\ref{subsec6:simulation_CKL}, these random elements are generated by their Karhunen-Lo\`eve expansions, therefore essentially enacting the Cram\'er-Karhunen-Lo\`eve representation \eqref{eq6:CKL_approx_truncated}.
On the other hand, the filtering specification discussed in Subsection~\ref{subsec6:simulation_filter} leverages the special structure of the filtered white noise spectral density operators to generate the random elements $\{Z_k\}$ efficiently. This approach is further tailored to simulation of FARFIMA processes in Subsection~\ref{subsec6:simulation_FARFIMA}.

Before moving on to the specifics, though, we establish that the sample generated by formula \eqref{eq6:simulation_iFFT} will indeed follow the correct dependence structure:
\begin{theorem}\label{theorem6:abstract_method}
Assume either of the two following conditions:
\begin{enumerate}[label=(\roman*)]
\item\label{item:theorem6:abstract_method:item_i}
The condition \eqref{eq6:weak_dependence_op_norm} holds and thus the spectral density operator $\{ \mathscr{F}^X_\omega \}_{\omega\in[0,2\pi]}$  exists in the sense \eqref{eq6:definition_spectral_density_operator}.
\item\label{item:theorem6:abstract_method:item_ii}
The weak spectral density operator $\mathscr{F}^X_\omega \in L^1([0,2\pi],\mathcal{L}_1(\HC))$ is continuous with respect to the norm $\|\cdot\|_1 $ on $(0,2\pi)$, and we additionally set $\mathscr{F}^X_0 = \mathscr{F}^X_{2\pi} = 0$.
\end{enumerate}
Then, the functional time series sample $X = \{X_t\}_{t=1}^T$ generated by \eqref{eq6:simulation_iFFT} is a real-valued stationary Gaussian time series of zero mean,
and asymptotically admits $\{\mathscr{F}_\omega\}$ as its spectral density operator when $T\to\infty$.
\end{theorem}

Theorem~\ref{theorem6:abstract_method} is proved in Appendix~\ref{subsec:proof_of_thm:abstract_method}.\\


Due to the periodicity of Fourier transform, the values $X_1$ and $X_T$ will tend to be similar which might be an undesirable trait, depending on the application. To overcome this artefact, \citet{mitchell1981generating,percival1993simulating} propose to simulate a sample of length $\tilde{T} = k T$ for some integer $k\geq 2$ and sub-sample a functional time series of length $T$.

\subsection{Simulation under Spectral Eigendecomposition Specification}
\label{subsec6:simulation_CKL}

Perhaps the most direct means to generate (approximate versions of) the random elements $\{Z_k\}$ considered in Step 1 of the algorithm introduced at the beginning of Section~\ref{sec6:simulation_in_spectral_domain} is by means of a finite rank approximation to the spectral density operator at the corresponding frequencies, appearing in the definition (see equation \eqref{eq6:simulate_Z_covariance_operator}). For a given rank, the optimal such approximation is obtained by truncating the eigenexpansion \eqref{eq6:spectral_density_operator_harmonic_decomposition} at that value, thus using a finite number of the harmonic eigenfunctions and corresponding eigenvalues to approximately generate $\{Z_k\}$.

Concretely, denoting $\{\lambda_n(\omega)\}_{n=1}^\infty$ and $\{\varphi_n(\omega)\}_{n=1}^\infty$ the harmonic eigenvalues and the harmonic eigenfunctions of the spectral density operator $\mathscr{F}^X_\omega$ at the frequency $\omega\in[0,2\pi]$, we may generate exact versions of $Z_k$ by setting
\begin{equation}\label{eq6:simulate_Z_CKL}
Z_k = \sum_{n=1}^\infty \sqrt{\lambda_n(\omega_k) } \varphi_n(\omega_k) \xi_n^{(k)}
\end{equation}
where $\{ \xi_n^{(k)}\}$ is an ensemble of i.i.d. standard Gaussian real-valued random variables.
The random elements defined by \eqref{eq6:simulate_Z_CKL} clearly satisfy the requirement \eqref{eq6:simulate_Z_covariance_operator}. In practice one has to truncate the series in \eqref{eq6:simulate_Z_CKL} at a finite level, say $N$. This truncation is optimal in terms of preserving the second order structure of the functional time series (Proposition~\ref{prop:optimality_CKL}) and requires only a low number of inexpensive operations. If we are to evaluate the functional time series $X$ on a spatial grid of $[0,1]$ at resolution $M\in\mathbb{N}$, the simulation requires $O(N M T + M T\log T)$ operations, provided we have direct access to the decomposition \eqref{eq6:spectral_density_operator_harmonic_decomposition}.  The $O(M T\log T)$ comes from the inverse fast Fourier transform \eqref{eq6:simulation_iFFT}.

When the decomposition \eqref{eq6:spectral_density_operator_harmonic_decomposition} is not directly available, as for example is the case for the \FARMA{} process with non-trivial autoregressive part, the evaluation of the spectral density operator \eqref{eq6:FARMA_spectral_density_operator} requires inversion of a bounded linear operator different at each frequency $\omega$. Unless a special structure of the autoregressive operator is assumed (e.g. as in Example~\ref{example6:long_range_FARFIMA}), the evaluation of this inversion is expensive. One could discretise the operator on a grid of $[0,1]^2$ and invert the resulting matrix, but this will become slow for dense grids, especially considering to do it for each frequency $\omega_k,\,k=1,\dots,T/2,T$. Moreover, to obtain the harmonic eigenvalues and eigenfunctions \eqref{eq6:spectral_density_operator_harmonic_decomposition} one would need to perform the eigendecomposition at each frequency  $\omega_k$ which is also slow for large matrices. These operations, if performed on a spatial grid of resolution $M\times M$, require $O(M^3)$ operations, bringing the overall cost to $O(M^3 T + M T\log T)$. This can be reduced by calling a truncated eigendecomposition algorithm instead, e.g. the truncated singular value decomposition (SVD) algorithm, and evaluating only $N < M$ eigenfunctions. This yields computational gains when $N\ll M$,  namely reducing the complexity of the said operations from $O(M^3)$ to $O(NM^2)$, and the overall cost to $O(N M^2 T + M T\log T)$.

Though the simulation cost is high when the decomposition \eqref{eq6:spectral_density_operator_harmonic_decomposition} is not directly available, the approach still constitutes a general method to simulate a functional time series with arbitrary spectrum. Example~\ref{example6:custom_karhunen_loeve} illustrates a functional time series whose dynamics are defined through its Cram\'er-Karhunen-Lo\`eve expansion where we show that simulation is possible even when we do not leverage our knowledge of this expansion, but rather calculate it numerically.

Finally, it is worth remarking that even though the functions $\{\varphi_n(\omega)\}_{n=1}^\infty$ appearing in \eqref{eq6:spectral_density_operator_harmonic_decomposition} are orthonormal for each $\omega\in[0,2\pi]$, orthonormality is not \emph{required} for the correct simulation of $Z_k$'s by \eqref{eq6:simulate_Z_CKL}. In other words, a practitioner can specify a spectral density operator by a sum similar to \eqref{eq6:spectral_density_operator_harmonic_decomposition} without insisting on using orthonormal functions, and still achieve rapid simulation in the spectral domain.

\subsection{Simulation under Filtering Specification.}
\label{subsec6:simulation_filter}

The second implementation of Step 1 of the abstract algorithm introduced at the beginning of Section~\ref{sec6:simulation_in_spectral_domain} leverages a set-up where a white noise with covariance operator $\mathcal{S}$ is plugged into a filter with given frequency response function $\Theta(\omega)$ in which case the spectral density operator is given directly by the formula
\begin{equation}\label{eq6:simulation_filter_spec_density}
\mathscr{F}^X_\omega = \frac{1}{2\pi} \Theta(\omega) \mathcal{S} \Theta(\omega)^*, \qquad\omega\in[0,2\pi],
\end{equation}
where $\mathcal{S}$ is a positive-definite self-adjoint trace class operator and $\Theta : [0,2\pi] \to \mathcal{L}(\HC)$, i.e. $\Theta(\omega)$ is a bounded linear operator on $\HC$ for each $\omega\in[0,2\pi]$. We only require that
$$ \int_0^{2\pi} \left\| \Theta(\omega) \right\|^2_{\mathcal{L}(\HC)} \D\omega <\infty$$
and $\Theta(\omega)g = \overline{\Theta(2\pi-\omega)(g)}$ for $\omega\in[0,\pi]$ and $g\in\HC$,
which implies that $\{X\}$ is a stationary mean-zero functional time series with the weak spectral density operator $\mathscr{F}^X_\omega \in L^1([0,2\pi],\mathcal{L}_1(\HC))$.

The operator $\mathcal{S}$, being a positive-definite self-adjoint trace class operator, admits the decomposition
\begin{equation}\label{eq6:simulation_filtered_S_decomposition}
\mathcal{S} = \sum_{n=1}^\infty \eta_n e_n \otimes e_n
\end{equation}
where $\{\eta_n\}$ are the eigenvalues and $\{e_n\}$ are the eigenfunctions of $\mathcal{S}$.

We may simulate real random elements $\{Y_k\}$ by setting
\begin{equation}\label{eq6:simulation_filtered_Y_k}
\sum_{n=1}^\infty \sqrt{\eta_n} e_n \tilde{\xi}_n^{(k)}
\end{equation}
with an ensemble $\{\tilde{\xi}_n^{(k)}\}$ of i.i.d. standard Gaussian random variables. In reality, the sum \eqref{eq6:simulation_filtered_Y_k} is truncated at some $N\in\mathbb{N}$.

If the decomposition \eqref{eq6:simulation_filtered_S_decomposition} is unknown, it can be numerically calculated by discretisation of the kernel corresponding to the operator $\mathcal{S}$ on the grid of $[0,1]^2$, say constituting an $M\times M$ matrix, and numerically calculating its eigendecomposition, in which case we may select $N=M$ eigenvalues.
The advantage of this approach over numerical evaluation of the spectral density operators at each $\omega$, performing the numerical eigendecomposition of each spectral density operator, and applying the Cram\'er-Karhunen-Lo\`eve-based simulation algorithm presented in Subsection~\ref{subsec6:simulation_CKL} is that the filtered white noise approach requires only one runtime of this expensive step.

Having defined the random elements $\{Y_k\}$ by \eqref{eq6:simulation_filtered_Y_k}, we define the elements $\{Z_k'\}$ in the notation of the algorithm presented at the beginning of Section~\ref{sec6:simulation_in_spectral_domain} by putting
\begin{equation}\label{eq6:simulation_filtered_Z_k}
Z_k' = \frac{1}{\sqrt{2\pi}} \Theta(\omega_k) Y_k, \qquad k=1,\dots,T/2,T.
\end{equation}
Such $\{Z_k'\}$ obviously satisfy \eqref{eq6:simulate_Z_covariance_operator}.

If the decomposition \eqref{eq6:simulation_filtered_S_decomposition} is unknown and we opt to numerically evaluate it on a grid of size $M$, the total computational complexity turns out to be $O(M^3 + M^2 T + M T\log T)$ where $O(M^2 T)$ comes from the matrix application \eqref{eq6:simulation_filtered_Z_k} and $O(M T\log T)$ from the inverse fast Fourier transform \eqref{eq6:simulation_iFFT}.

\subsection{Simulation under Linear Time Domain Specification}
\label{subsec6:simulation_FARFIMA}

One of the typical functional time series dynamics specifications is a linear process in the time domain. In this subsection we consider the flexible class of the \FARFIMA{} processes, one of the most general classes of such linear processes, and show how to generate their trajectories by spectral domain simulation methods.

The \FARFIMA{} processes, thanks to being defined as a linear filter of white noise, admit the spectral density operators of the form \eqref{eq6:simulation_filter_spec_density}. However, the application of the simulation algorithm presented in Subsection~\ref{subsec6:simulation_filter} requires the frequency response function $\Theta(\omega)$ to be readily available, which is not always the case: the \FARFIMA{} (or \FARMA{}) process with a non-degenerate autoregressive part admit the frequency response function given by the formula prompting operator inversion:
\begin{equation}\label{eq6:FARIMA_frequency_response_function}
\Theta(\omega) =
\mathscr{A}( e^{-\I\omega} )^{-1}
\mathscr{B}( e^{-\I\omega} ),\qquad
\omega\in[0,2\pi].
\end{equation}

Therefore a naive implementation would require inversion of the linear bounded operator $\mathscr{A}( e^{-\I\omega} )$ for each frequency $\omega$. It may very well happen that $\mathscr{A}( e^{-\I\omega} )$ has special structure, e.g. as is the case for the FARFIMA(1,d,0) process considered in Example~\ref{example6:long_range_FARFIMA}, in which case the inversion evaluation is rapid. In the general case, however, the inversion on a spatial domain discretisation would require $O(M^3)$ operations where $M$ is the discretisation resolution. Fortunately, there are two ways to avoid this computational cost:
\begin{itemize}
\item A \textit{fully spectral approach} which consists in the efficient evaluation of \eqref{eq6:simulation_filtered_Z_k}. The discretization of this formula for the \FARFIMA{} process involves evaluation of
\begin{equation}\label{eq6:simulation_FARIMA_get_V_k}
Z_k = \frac{[2\sin(\omega/2)]^{-d}}{\sqrt{2\pi}} \mathbf{A}( e^{-\I\omega_k} )^{-1} \mathbf{B}( e^{-\I\omega_k} ) Y_k
\end{equation}
where the matrices $\mathbf{A}( e^{-\I\omega_k} )$ and $\mathbf{B}( e^{-\I\omega_k} )$ are the discretizations of $\mathscr{A}( e^{-\I\omega_k} )$ and $\mathscr{B}( e^{-\I\omega_k} )$ respectively.
The numerical evaluation of \eqref{eq6:simulation_FARIMA_get_V_k} requires solving the matrix equation with the matrix $\mathbf{A}( e^{-\I\omega_k} )$ and the right-hand side vector of $ \mathbf{B}( e^{-\I\omega_k} ) Y_k$, thus resulting in $O(M^2)$ complexity, as opposed to the $O(M^3)$ complexity of matrix inversion.


\item A \textit{hybrid simulation approach}, where we simulate the \FARFIMA{} processes by simulating the corresponding FARFIMA$(0,d,q)$ process in the spectral domain and then applying the autoregressive recursion in the time-domain. Concretely, we:
\begin{enumerate}
\item Choose a burn-in length $\tilde{T}$, and simulate a FARFIMA$(0,d,q)$ process with degenerate autoregressive part, denoted as $X'_1,\dots,X'_{T+\tilde{T}}$, by the means of the tools in Subsection~\ref{subsec6:simulation_filter}. Such a functional time series admits the spectral density operator
$$ \mathscr{F}_\omega^{X'} = \frac{\left[2\sin(\omega/2)\right]^{-2d}}{2\pi} \mathscr{B}( e^{-\I\omega} ) \mathcal{S} \mathscr{B}( e^{\I\omega} )^* $$
whose corresponding frequency response function $\Theta(\omega) = [2\sin(\omega/2)]^{-d} \mathscr{B}( e^{-\I\omega} )$ can be evaluated fast.
\item Set $X_1,\dots, X_p = 0$ and run the recursion
$$ X_t = \mathcal{A}_1 X_{t-1} + \dots + \mathcal{A}_p X_{t-p} + X'_t, \qquad t = p+1,\dots,T+\tilde{T}.$$
\item Discard the first $\tilde{T}$ values of $X_1,\dots,X_{T+\tilde{T}}$ and keep only the last $T$ elements.
\end{enumerate}
\end{itemize}

Both the fully spectral and the hybrid implementations involve the numerical eigendecomposition of the noise covariance operator $\mathcal{S}$, incurring an $O(M^3)$ computation cost, the applications of matrices on vectors or solving linear equations, yielding $O(M^2 T)$ operations, and the inverse fast Fourier transform at each point of the discretisation with the $O(M T\log T)$ complexity. Thus the total computational complexity is $O(M^3 + M^2T + M T\log T)$.
Nevertheless, even though the application of a matrix on a vector has the same complexity as solving a linear system of equations, the constant hidden in the ``$O$" is different and the hybrid simulation method is faster than the fully spectral approach, which requires the solution of linear systems at each frequency, as the simulation study in Example~\ref{example6:FARMA_lowrank} demonstrates.

\section{Examples and Numerical Experiments}
\label{sec6:examples}

This section presents three examples of functional time series specified according in various ways, similarly to the last three section. Thus, the spectral density operator may be directly or indirectly defined, depending on the scenario. The examples are accompanied by a small simulation study assessing the simulation speed and the simulation accuracy by comparing the lagged autocovariance operators of the simulated processes with the ground truth. The purpose of the simulation study is to illustrate the performance of the method in terms of speed and accuracy, and draw some qualitative conclusions about the choice of methods and parameters, rather than to provide with an extensive quantitative comparison.

A parallel objective is to provide code that is accessible (Section~\ref{sec:code_availability}), simple to run, and easy to tailor for custom-defined spectral density operators used in functional time series research.

\subsection{Specification by Spectral Eigendecomposition}
\label{example6:custom_karhunen_loeve}
Consider the spectral density operator defined by its eigendecomposition
\begin{align}
\label{eq6:custom_CKL_spec_density_operator_sum}
\mathscr{F}^X_\omega &= \sum_{n=1}^\infty \lambda_n(\omega) \varphi_n(\omega)\otimes \varphi_n(\omega), \qquad \omega\in[0,2\pi],\\
\nonumber
\lambda_n(\omega) &= \frac{1}{ (1-0.9 \cos(\omega)) \pi^2 n^2}, \qquad \omega\in[0,2\pi],\\
\nonumber
\left(\varphi_n(\omega)\right)(x) &=
\begin{cases}
\sqrt{2} \sin(  n (\pi \delta_{\omega/\pi}(x) ),& \quad x\in[0,1],\quad \omega\in[0,\pi], \\
\sqrt{2} \sin(  n (\pi \delta_{-\omega/\pi}(x)) ),& \quad x\in[0,1],\quad \omega\in(\pi,2\pi],
\end{cases} 
\end{align}
where  $$\delta_a(\cdot) = x-a \mod 1$$ is the periodic shift by $a\in\mathbb{R}$ with ``mod" denoting the modulo operation, the remainder after the division.
Under such definition, which guarantees that $\delta_a(x) \in [0,1]$, the harmonic eigenfunctions at distinct frequencies are phase-shifted versions of each other. It turns out that the spectral density operator given by the sum \eqref{eq6:custom_CKL_spec_density_operator_sum} can be expressed in closed analytical form, as an integral operator with kernel
\begin{equation}
\label{eq6:custom_CKL_spec_density_kernel}
f^X_\omega(x,y) =
\begin{cases}
\frac{1}{ (1-0.9 \cos(\omega))} K_{BB}( \delta_{\omega/\pi}(x), \delta_{\omega/\pi}(x) ),& \omega\in[0,\pi],\\
\frac{1}{ (1-0.9 \cos(\omega))} K_{BB}( \delta_{-\omega/\pi}(x), \delta_{-\omega/\pi}(x) ),& \omega\in(\pi,2\pi].
\end{cases}
\end{equation}
where $K_{BB}(\cdot,\cdot)$ is the covariance kernel of Brownian bridge  \citep{deheuvels2003} defined as
$$
K_{BB}(x,y) = \min(x,y) - xy,\qquad x,y\in[0,1].
$$

Figure~\ref{fig6:custom_CKL_trajectory} illustrates the simulated trajectories with varying number of the harmonic principal components $N$ used in the truncation of the sum \eqref{eq6:simulate_Z_CKL} when simulating by the means presented in Subsection~\ref{subsec6:simulation_CKL}.

\begin{figure}
\centering
\includegraphics[width=1\textwidth]{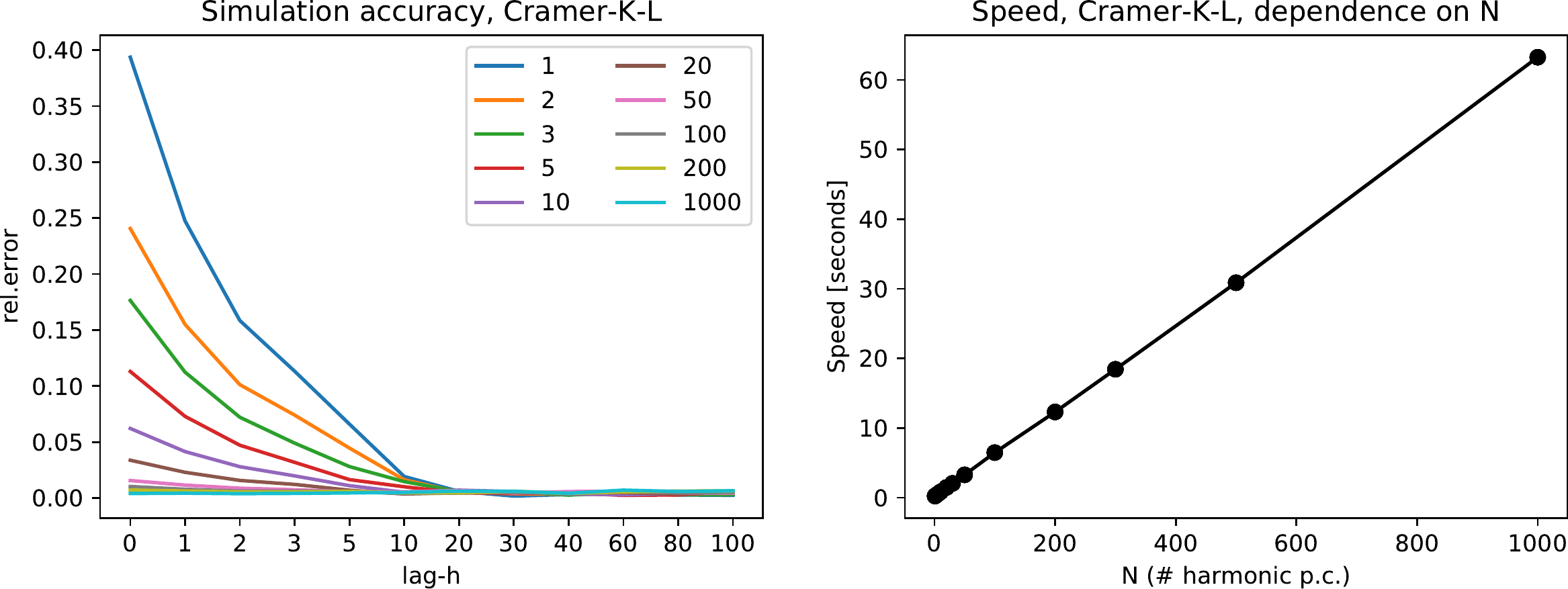}
\caption[Simulation accuracy and speed of the Cram\'er-Karhunen-Lo\`eve method for the simulation in Example~\ref{example6:custom_karhunen_loeve}]{
The simulation accuracy \eqref{eq6:RMSE_simulation} and speed of the process defined in Example~\ref{example6:custom_karhunen_loeve}.
\textbf{Left:} The simulation accuracy for lag-$h$ autocovariance operator with varying number of harmonic principal components used $N\in\{1,2,3,5,10,20,50,100,200,1000\}$ visualised as a function of the lag $h\in\{0,1,2,3,5,10,20,30,40,60,80,100\}$. The sample size parameters are set $T=1000$ and $M=1001$.
\textbf{Right:} The simulation speed as a function of $N$ with fixed $T=1000$ and $M=1001$.
}
\label{fig6:custom_CKL_accuracy}
\end{figure}

\begin{figure}
\centering
\includegraphics[width=1\textwidth]{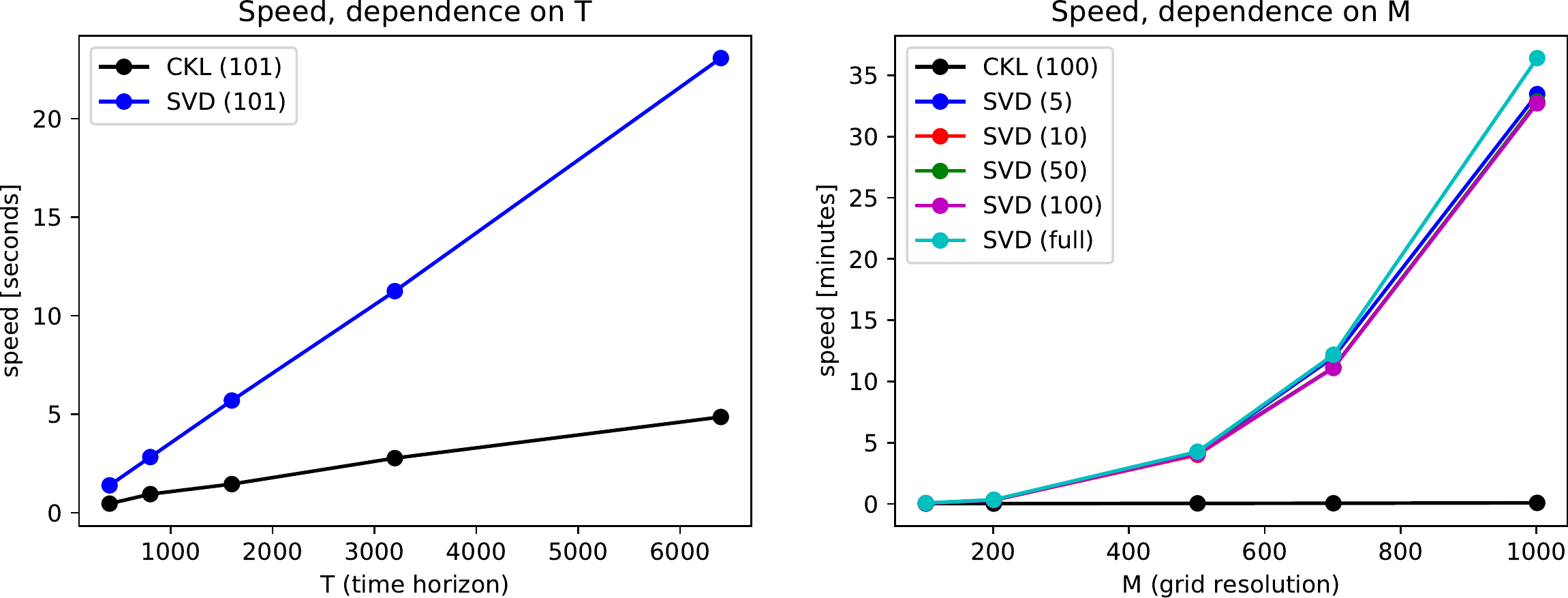}
\caption[Simulation speed for the simulation in Example~\ref{example6:custom_karhunen_loeve}]{
The simulation speed of the process defined in Example~\ref{example6:custom_karhunen_loeve}.
\textbf{Left:} The dependence on varying the time horizon $T\in\{400,800,1600,3200,6400\}$ while setting the spatial resolution $M=101$. Both the simulation using the known Cram\'er-Karhunen-Lo\`eve expansion (\textsc{CKL}) and the method calculating this decomposition by the \textsc{SVD} algorithm use $N=101$ eigenfunctions.
\textbf{Right:} The dependence on varying $M\in\{101,201,501,701,1001\}$ while setting $T=1000$. The simulation using the known Cram\'er-Karhunen-Lo\`eve expansion (\textsc{CKL}) uses $100$ eigenfunctions while the numerical \textsc{SVD ($N$)} decomposition finds $N\in\{5,10,50,100\}$ leading eigenfunctions (the lines mostly overlap each other) or all of them $N=M$ for \textsc{SVD (full)}. The \textsc{CKL} method has the running time below 0.1~minutes (6~seconds) even for $M=1001$.
}
\label{fig6:custom_CKL_speed}
\end{figure}

\begin{figure}
\centering
\includegraphics[width=1\textwidth]{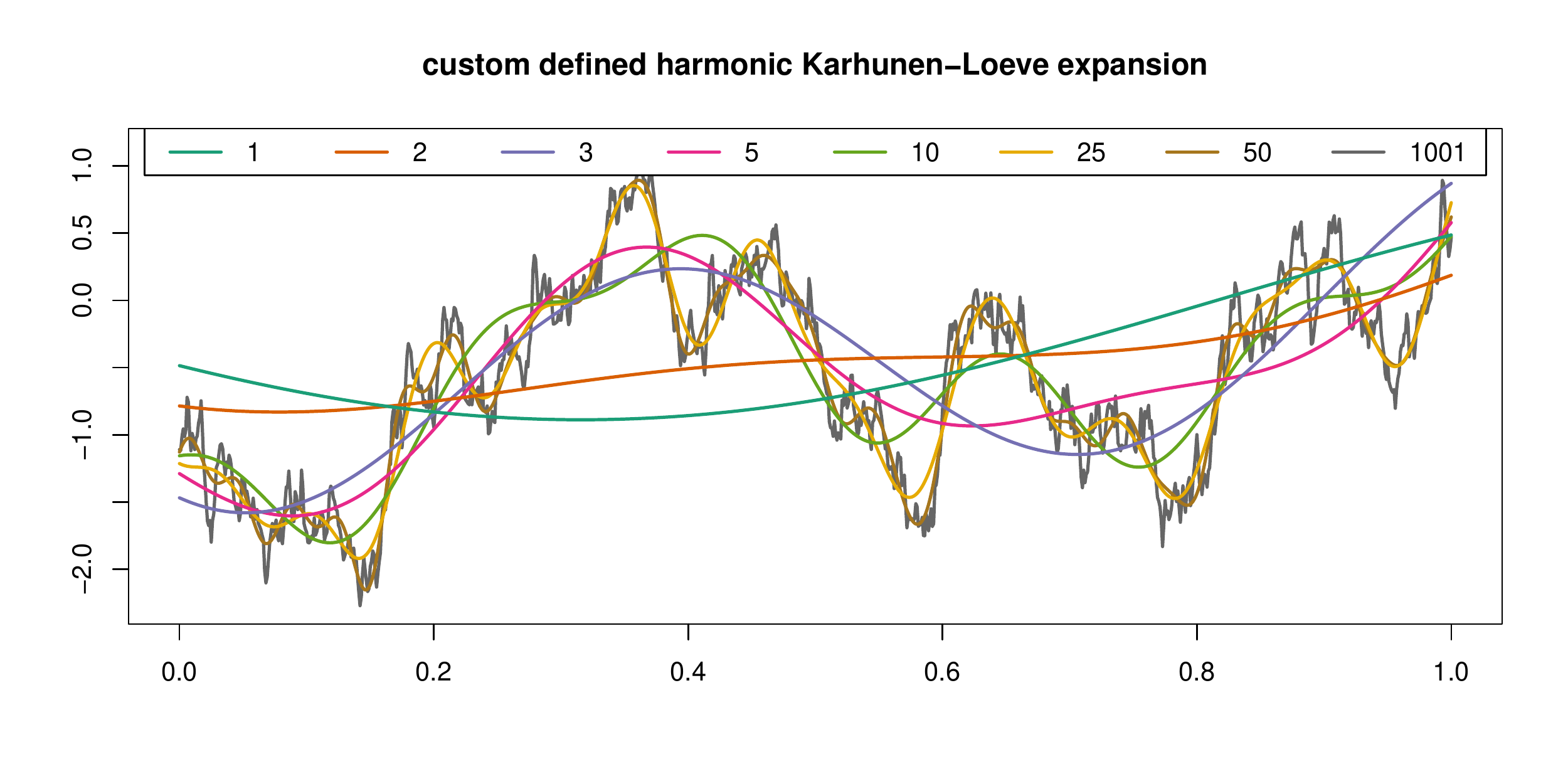}
\caption[Simulated trajectories of the process defined in Example~\ref{example6:custom_karhunen_loeve}]{
Sample trajectories $X_1(\cdot)$ of the process defined in Example~\ref{example6:custom_karhunen_loeve} with varying number of harmonic principal components $N$ chosen in the truncation of \eqref{eq6:simulate_Z_CKL}. Simulated with $T=100$ and the grid resolution $M=1001$.
}
\label{fig6:custom_CKL_trajectory}
\end{figure}

In order to assess the simulation accuracy we opt to: simulate $I=1000$ independent realisations of the process $\{X^{(1)}_t\}_{t=1}^T,\dots,\{X^{(I)}_t\}_{t=1}^T$; evaluate its empirical autocovariance operators $\hat{\mathscr{R}}^X_{h,[i]}$ for each $i=1,\dots,I$ and some lags $h$; and define the average empirical autocovariance operator $\overline{\mathscr{R}^X_h} = \frac{1}{I}\sum_{i=1}^I \hat{\mathscr{R}}^X_{h,[i]}$. We then compare this with the true covariance operator $\mathscr{R}^X_h$ by calculating
\begin{equation}\label{eq6:RMSE_simulation}
rel.error(h) = \frac{\left\| \overline{\mathscr{R}^X_h} - \mathscr{R}^X_h \right\|_1}{\left\| \mathscr{R}^X_0 \right\|_1}, \qquad\text{for some lags}\,\, h.
\end{equation}
The true autocovariance operators $\mathscr{R}^X_h$ were calculated by numerically integrating \eqref{eq6:spectral_density_operator_inverse_formula}.

Figure~\ref{fig6:custom_CKL_accuracy} the manner of error decay as $N\to\infty$ and the number of harmonic components $N=100$ seems to be satisfactory. The relative simulation errors for $N>100$ seem to be dominated by the random component of \eqref{eq6:RMSE_simulation} rather than the simulation error itself.
We note that the spectral density operator \eqref{eq6:custom_CKL_spec_density_operator_sum} is non-differentiable near the spatial diagonal, and consequently features a relatively slow (quadratic) decay of its eigenvalues. It thus represents one of the more challenging cases one might wish to simulate from in an FDA context: functional data analyses typically feature smooth curves and differentiable corresponding operators, including spectral density operators, admitting a faster quicker eigenvalue requiring $N\ll 100$ eigenfunctions to capture a substantial amount of their variation.

Figure~\ref{fig6:custom_CKL_speed} presents the simulation speed results with varying sample size parameters: the time horizon $T$ and the spatial resolution $M$. We compared the simulation using the known Cram\'er-Karhunen-Lo\`eve decomposition \eqref{eq6:custom_CKL_spec_density_operator_sum} with the method finding this decomposition numerically starting from the kernel 
\eqref{eq6:custom_CKL_spec_density_kernel}. Such method finds the harmonic eigendecomposition using the (truncated) SVD algorithm applied to discretization of \eqref{eq6:custom_CKL_spec_density_kernel}. Figure~\ref{fig6:custom_CKL_speed} shows that such routine can become very costly for higher spatial resolutions $M$, but if no other method is available, the method still constitutes an general approach how to simulate process with any dynamics structure defined through weak spectral density operators.

\subsection{Long-range Dependent \FARFIMA{} Process}
\label{example6:long_range_FARFIMA}

The next example is sourced from the work of \citet{li2019long} and \citet{shang2020comparison} on long-rang dependent functional time series. They consider the FARFIMA(1,0.2,0) process defined by \eqref{eq6:FARIMA_def} with the autoregressive operator $\mathcal{A}_1$ and the innovation covariance operator $\mathcal{S}$ defined as integral operators with respective kernels
\begin{align}
\label{eq6:example_FARIMA_def_A1}
A_1(x,y) &= 0.34 \exp\left\{ (x^2+y^2)/2 \right\}, \qquad x,y\in[0,1],\\
\label{eq6:example_FARIMA_def_S}
S(x,y) &= \min(x,y), \qquad x,y\in[0,1],
\end{align}
depicted in Figure~\ref{fig6:FARIMA_kernels}. Recall that $S(x,y)= \min(x,y)$ is the covariance kernel of the standard Brownian motion on $[0,1]$.
Because $d=0.2 > 0$, the process exhibits long-rang dependence \citep{li2019long}.

The constant $0.34$ ensures that condition \eqref{eq6:FARMA_condition} is satisfied, and thus the process is stationary and admits a weak spectral density operator (Theorem~\ref{thm:FARFIMA-part-ii}) given by
\begin{equation}\label{eq6:example_FARIMA_spec_density}
\mathscr{F}^X_\omega = \frac{\left[2\sin(\omega/2)\right]^{-2d}}{2\pi}
\left( \Id - \mathcal{A}_1 e^{-\I\omega} \right)^{-1}
\mathcal{S}
\left( \Id - \mathcal{A}_1^* e^{\I\omega} \right)^{-1},
\qquad\omega\in[0,2\pi].
\end{equation}
In fact, the operator $\mathcal{A}_1$ is of rank 1 and can be written as $\mathcal{A}_1 = -0.34 g\otimes g$ with $g(x)=\exp( x^2/2 ),\,x\in[0,1]$. This fact hugely simplifies the evaluation of \eqref{eq6:example_FARIMA_spec_density} because the inversion of the autoregressive part can be written by the Sherman–Morrison formula as
\begin{equation}\label{eq6:sherman-morrison}
\left( \Id - \mathcal{A}_1 e^{-\I\omega} \right)^{-1} = \Id +
\frac{0.34 e^{-\I\omega}}{1-0.34 e^{-\I\omega} \|g\|^2_{\LtwoR}} g\otimes g,
\qquad\omega\in[0,2\pi],
\end{equation}
thus allowing for fast evaluation. Further computation gains, though less considerable, are made by using the Mercer decomposition of the Brownian motion covariance kernel \citep{deheuvels2003}
\begin{equation}\label{eq6:brownian_motion_KL}
S(x,y) = \sum_{n=1}^\infty \frac{1}{\left[(n-0.5)\pi\right]^2} \sqrt{2} \sin\left\{ (n-0.5) \pi x \right\} \sqrt{2} \sin\left\{ (n-0.5) \pi y \right\},\qquad x,y\in[0,1],
\end{equation}
instead of numerical evaluation on a grid followed by an SVD decomposition.

In what follows, we consider the following implementations the spectral and time-domain, and hybrid simulation methods:

\begin{itemize}
\item \textsc{spectral (bm)}: This method uses the known Mercer decomposition of the Brownian motion (\textsc{bm}) kernel \eqref{eq6:brownian_motion_KL} and simulates the process in the spectral domain using the method of Subsection~\ref{subsec6:simulation_filter}
with the help of the Sherman-Morrison formula \eqref{eq6:sherman-morrison}.
\item \textsc{hybrid (bm)}: This method again uses the known Mercer decomposition of the Brownian motion (\textsc{bm}) kernel \eqref{eq6:brownian_motion_KL} and simulates the FARFIMA$(0,d,0)$ process and then applies the autoregressive recustion in the time-domain as explained in Subsection~\ref{subsec6:simulation_FARFIMA}, thus constituting a \textsc{hybrid} simulation method combining spectral and time-domain.
\item \textsc{spectral (svd)}, \textsc{hybrid (svd)}: These method correspond to \textsc{spectral (bm)} and \textsc{hybrid (bm)} but the Mercer decomposition of the Brownian motion kernel is calculated numerically using the \textsc{svd} algorithm.
\item \textsc{temporal}: We use the original code by \citet{li2019long} available in the on-line supplement of their article and treat is as the benchmark for comparison with our spectral simulation methods. They simulate the realisations of the process by discretising the space domain $[0,1]$ and evaluating the integral operator $\mathcal{A}_1$ as a sum on this grid. Moreover, they perform the fractional integration \eqref{eq6:FARIMA_def} by analytically calculating the filter coefficients in the time-domain and thus expressing the process as FMA($\infty$), the functional moving average process of infinite order. Details on the FMA($\infty$) representation can be found in \citet{li2019long,hosking1981fractional}. The computational complexity of this method is $O(M^2 T^2)$.
\end{itemize}

In order to assess the simulation accuracy we opt to simulate $I=100$ independent realisations, and compare the mean empirical autocovariance operators \eqref{eq6:RMSE_simulation} with the true autocovariance operator for varying $T\in\{400,800,1600,3200,6400\}$ and $M\in\{101,201,501,1001\}$.
We simulate the process with varying parameter $T$, the time horizon of the simulation, as well as varying spatial resolution $M$, based on a regular grid $\{x_m = (m-1)/(M-1)\}_{m=1}^M \subset [0,1]$.  
The simulation accuracy error, reported in Figure~\ref{fig6:FARFIMA} (in Appendix~\ref{sec:supplementary_figures}), is negligible for all the simulation methods and \eqref{eq6:RMSE_simulation} is dominated rather by the random component,  which is higher for smaller $T$.

Figures~\ref{fig6:FARIMA_speed} summarise how fast the different simulation methods were. It is obvious that the simulation by the \textsc{temporal} method used by \citet{li2019long} scales badly in $T$, while the other methods are linear in $T$, performing significantly better.
On the other hand, the \textsc{spectral (bm)}, \textsc{hybrid (bm)}, and \textsc{temporal} methods taking advantage of the innovation error covariance eigendecomposition have complexity dominated by $O(M^2 T)$ and scale similarly. The \textsc{spectral (svd)} and \textsc{hybrid (svd)} methods require a further $O(M^3)$ operations for the SVD algorithm and this contribution becomes visible for $M\in\{501,1001\}$.

\begin{figure}
\centering
\includegraphics[width=1\textwidth]{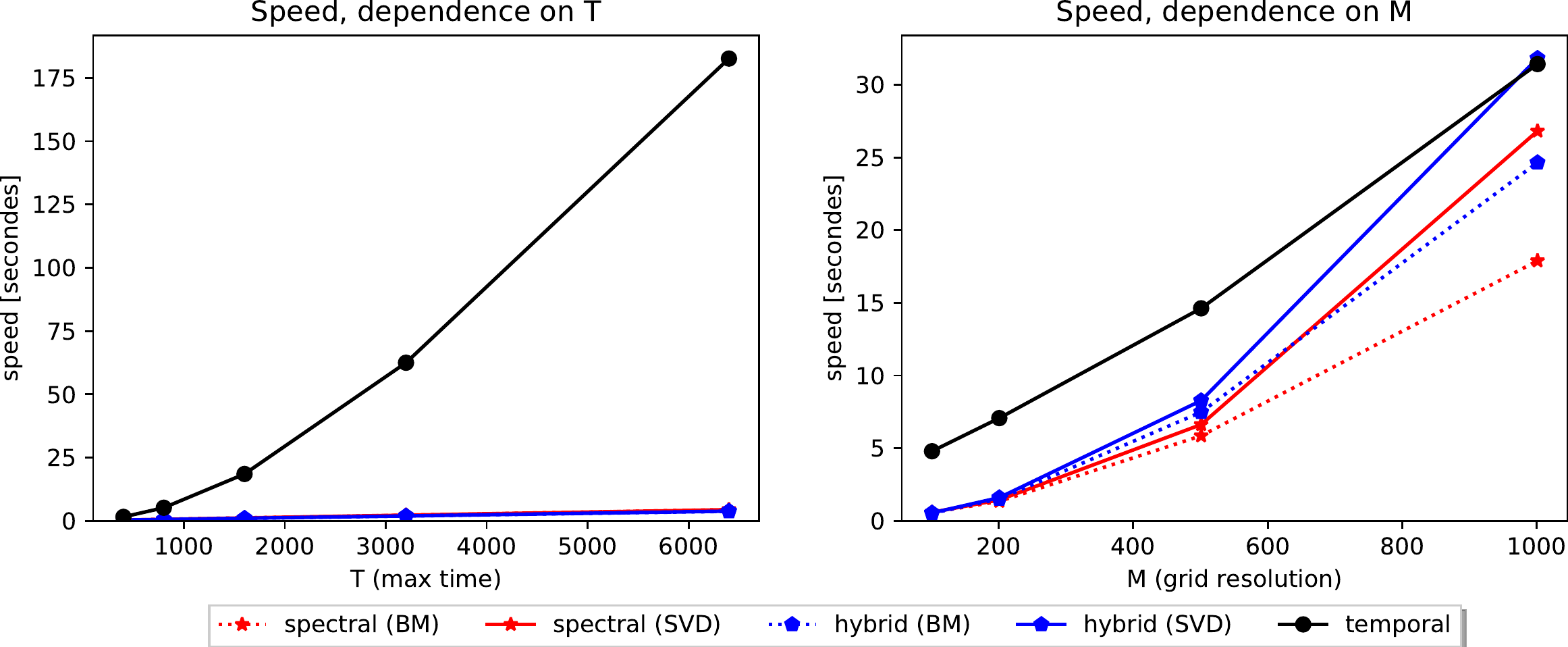}
\caption[The simulation speed for the FARFIMA(1,0.2,0) in Example~\ref{example6:long_range_FARFIMA}]{
The dependence of the \textbf{simulation speed} for the long-range dependent FARFIMA(1,0.2,0) process defined in Example~\ref{example6:long_range_FARFIMA} on the simulation parameters. 
\textbf{Left:} The simulation speed for varying time horizon $T\in\{400,800,1600,3200,6400\}$ with the spatial resolution is set $M=101$.
\textbf{Right:} The dependence of the simulation speed  on the grid size $M\in\{101,201,501,1001,\}$ with $T=800$.
}
\label{fig6:FARIMA_speed}
\end{figure}

\subsection{\FARMA{} Process with Smooth Parameters}
\label{example6:FARMA_lowrank}

\begin{figure}[ht]
\centering
\includegraphics[width=1\textwidth]{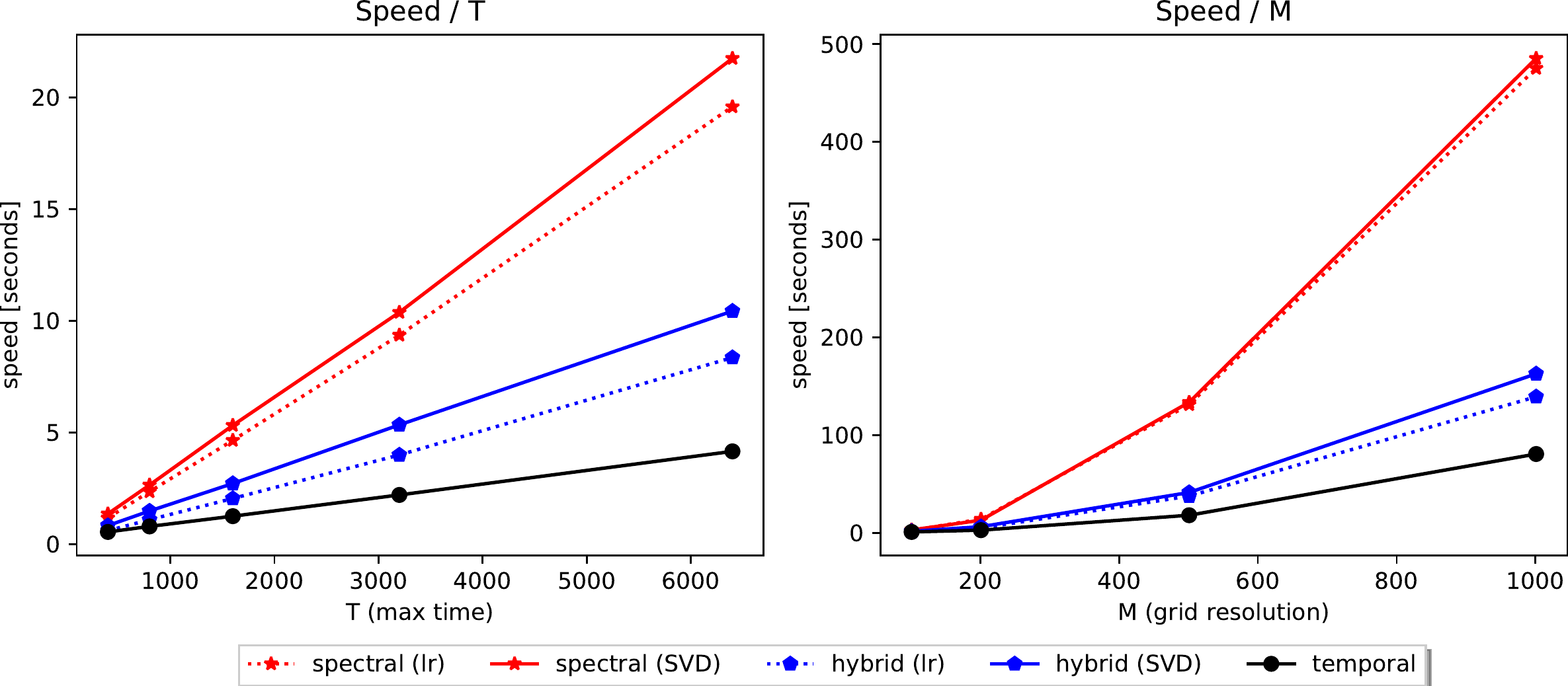}
\caption[The simulation speed for the FARMA(4,3) in Example~\ref{example6:FARMA_lowrank}]{
The dependence of the \textbf{simulation speed} for the FARMA(4,3) process defined in Example~\ref{example6:FARMA_lowrank} on the simulation parameters.
\textbf{Left:} The simulation speed for varying time horizon $T\in\{400,800,1600,3200,6400\}$ with the spatial resolution is set $M=101$.
\textbf{Right:} The dependence of the simulation speed  on the grid size $M\in\{101,201,501,1001\}$ with $T=800$.
}
\label{fig6:FARMA_speed}
\end{figure}

In this example we consider the FARMA(4,3) process \eqref{eq6:FARMA_def} with the autoregressive operators $\mathcal{A}_1,\dots,\mathcal{A}_4$, the moving average operators $\mathcal{B}_1,\dots,\mathcal{B}_3$, and the innovation covariance operator $\mathcal{S}$ defined as integral operators with kernels
\begin{align*}
A_1(x,y) &= 0.3 \sin(x-y), & B_1(x,y) &= x+y, \\
A_2(x,y) &= 0.3 \cos(x-y), & B_2(x,y) &= x, \\ 
A_3(x,y) &= 0.3 \sin(2x),  & B_3(x,y) &= y, \\
A_4(x,y) &= 0.3 \cos(y),   &
\end{align*}
and
\begin{align}\label{eq6:FARMA_lowrank_def_S}
S(x,y) =
           &\sin(  2\pi x )\sin(  2\pi y ) +\\ \nonumber
    + 0.6  &\cos(  2\pi x )\cos(  2\pi y ) +\\ \nonumber
    + 0.3  &\sin(  4\pi x )\sin(  4\pi y ) +\\ \nonumber
    + 0.1  &\cos(  4\pi x )\cos(  4\pi y ) +\\ \nonumber
    + 0.1  &\sin(  6\pi x )\sin(  6\pi y ) +\\ \nonumber
    + 0.1  &\cos(  6\pi x )\cos(  6\pi y ) +\\ \nonumber
    + 0.05 &\sin(  8\pi x )\sin(  8\pi y ) +\\ \nonumber
    + 0.05 &\cos(  8\pi x )\cos(  8\pi y ) +\\ \nonumber
    + 0.05 &\sin( 10\pi x )\sin( 10\pi y ) +\\ \nonumber
    + 0.05 &\cos( 10\pi x )\cos( 10\pi y ), \qquad x,y\in[0,1].
\end{align}
These are depicted in Appendix~\ref{sec:supplementary_figures}, Figure~\ref{fig6:FARMA_kernels}.
The constant $0.3$ guarantees stationarity of the process, hence it admits the spectral density \eqref{eq6:FARMA_spectral_density_operator}. Figure~\ref{fig6:FARMA}, included in Appendix~\ref{sec:supplementary_figures}, confirms that all the simulation methods approximate well the simulated process as the relative simulation error metric is affected more by the stochastic component. Figure~\ref{fig6:FARMA_speed} presents the simulation speed comparison between the spectral domain methods and the time-domain autoregressive recursion approach (\textsc{temporal}). The four considered spectral domain methods include:
\begin{itemize}
\item \textsc{spectral (lr)} 
This method uses the eigendecomposition \eqref{eq6:FARMA_lowrank_def_S} of the innovation noise covariance kernel. The simulation is conducted fully in the spectral domain as explained in Subsection~\ref{subsec6:simulation_FARFIMA}.
\item \textsc{hybrid (lr)}: This method uses the eigendecomposition \eqref{eq6:FARMA_lowrank_def_S} of the innovation noise covariance kernel, simulates the corresponding moving average process in the spectral domain and applies the autoregressive part in the time-domain as explained in Subsection~\ref{subsec6:simulation_FARFIMA}.
\item \textsc{spectral (svd)}, \textsc{hybrid (svd)}: As above, but the eigendecomposition of $S(x,y)$ is calculated numerically by the SVD algorithm.
\end{itemize}
Even though the time complexity, which is dominated by the term $O(M^2T)$, of the spectral domain simulation method matches the time complexity of the \textsc{temporal} domain approach with $O(M^2T)$ complexity, the results presented in Figure~\ref{fig6:FARMA_speed} show that the simulation of the \FARMA{} process in the spectral domain, requiring solving matrix equation at each frequency, as well as the hybrid simulation are slower than the \textsc{temporal} approach.

The low-rank definition of \eqref{eq6:FARMA_lowrank_def_S} does not yield any computational speed-up compared to infinite rank covariance kernels (such as the Brownian motion kernel in Example~\ref{example6:long_range_FARFIMA}). The purpose of such a definition is to allow for easy modification of the code if one wishes to specify the process via its harmonic eigenfunctions.

\section{General Recommendations for Simulations}
\label{sec:general_redommendations}

Our methodology provides a general purpose toolbox for simulating stationary (Gaussian) functional time series, leveraging their spectral representation. The high-level skeleton outlined at the beginning of Section~\ref{sec6:simulation_in_spectral_domain} essentially reduces the problem to simulating a finite ensemble of independent random elements, and then applying the inverse fast Fourier transform. The generation of this i.i.d. ensemble depends on how one chooses to carry out discretisation and/or dimension reduction. We have demonstrated how knowledge of additional structure can significantly speed up the computations.

\medskip
\noindent Some take-away messages and recommendations are as follows. 
\begin{itemize}
\item \textbf{Simulation of functional time series specified through their spectral density operator.}
To date, this problem had not been addressed, presumably because the assessment of the functional time series methods has traditionally been done based on simulation of functional \emph{linear} processes. Key methods pertaining to regression and prediction, however, present performance tradeoffs that depend on the frequency domain properties, rather than the time domain properties of the time series \citep{hormann2015dynamic,hormann2015estimation,hormann2018testing,zhang2016white,tavakoli2016detecting,pham2018methodology,rubin2019functional,rubin2020sparsely}.
One then wishes to simulate from a spectrally specified functional time serirs.  More generally, our method can in principle be applied to any stationary model, linear or nonlinear, going well beyond the classical families of functional \FARMA{} or \FARFIMA{} processes, provided the process admits a weak spectral density operator.

The method is fast and produces accurate results when the process is spectrally specified, courtesy of the Cram\'er-Karhunen-Lo\`{e}ve expansion (Subsection~\ref{subsec6:simulation_CKL}) which is provably the optimal way to carry out dimension reduction. Excellent performance can also be expected when the dynamics of a functional time series are specified by means of white noise filtering (Subsection~\ref{subsec6:simulation_filter}). 
 For a general specification, the spectral domain simulation method of Subsection~\ref{subsec6:simulation_CKL} still provides means how to simulate arbitrary functional time series. If the Cram\'er-Karhunen-Lo\`{e}ve expansion is unknown, or a filtering representation is not available, the spectral density evaluation and the numeric eigendecomposition might require more time-consuming operations.  Still, the approach constitutes the only general purpose recipe, where no previous method was available.
 
\item \textbf{Simulation of \FARFIMA{} processes.}
The advantages of the spectral approach compared to time domain methods become quote considerable when dealing with processes that have an infinite order moving average representation, while having a simple formulation in the spectral domain. An important example being the \FARFIMA{} processes with $d>0$ (long memory process) or $d<0$ (anti-persistent) as the fractional integration is straightforward in the spectral domain while it produces an infinite order dependence in the time-domain. Example~\ref{example6:long_range_FARFIMA} showed how to efficiently and effortlessly simulate a long-range dependent FARFIMA process. Therefore we submit that the simulation of \FARFIMA{} processes with $d\neq 0$ is more accessible and easy to implement in the spectral domain.

\item \textbf{Simulation of \FARMA{} processes.} If one does specifically want to simulate a \FARMA{} processes, simulation in the time-domain is straightforward and fast. Still, our spectral domain simulation method matches the time complexity of the time domain methods in these cases.
The constant hidden in ``$O$", however, seems to be higher for the spectral domain methods, as Example~\ref{example6:FARMA_lowrank} confirms. One advantage that the simulation in the spectral domain attains over the time-domain, though, is that we do not need to worry about the burn-in to reach the stationary distribution. We tentatively conclude that if a practitioner wishes to simulate a \FARMA{} process, then both the time-domain and the spectral domain methods are equally applicable, though the time-domain simulation seems to be more straightforward to implement.
\end{itemize}

Overall  the presented methods provide a useful toolbox of simulation methods in the spectral domain which are fast and accurate, and allow for simulation of standard as well as unusual or ``custom defined" stationary time series defined through their weak spectral density operators. We hope that the accompanying code can be helpful for carrying out numerical experiments in future functional time series methodological research.


\section{Code Availability and \texttt{R} Package \texttt{specsimfts}}
\label{sec:code_availability}

To facilitate the implementation of spectral domain simulation methods introduced in this article, we have created an \texttt{R} package \texttt{specsimfts} available on GitHub at \url{https://github.com/tomasrubin/specsimfts}. The package includes the implementations of all the methods presented in this article as well as the examples considered in Section~\ref{sec6:examples} as demos that are easy to use and modify.

\appendix

\section{Proofs of Formal Statements}
\label{sec:proofs}

\subsection{Proof of Theorem~\ref{thm:FARMA-part-ii}}
\label{subsec:proof_of_thm:FARMA}

\begin{proof}
Denoting $\Delta$ to be the backshift operator, the equation \eqref{eq6:FARMA_def} can be rewritten as 
\begin{equation}\label{eq6:FARMA_filtrations}
\mathscr{A}(\Delta) X_t = \mathscr{B}(\Delta) \epsilon_t.
\end{equation}

We start with the analysis of the moving average part
\begin{equation}\label{eq6:ma_filtration}
\eta_t = \mathscr{B}(\Delta) \epsilon_t.
\end{equation}
The spectral density operator of the white noise process $\{\epsilon_t\}$ is trivially given by $\mathscr{F}^\epsilon_\omega = (2\pi)^{-1} \mathcal{S}$. 
The filter $\mathscr{B}(\Delta)$, whose filter coefficients are given by $\mathscr{B}(\Delta)_s = \mathcal{B}_s$ for $s=0,\dots,q$ and $\mathscr{B}(\Delta)_s = 0$ otherwise, defines the frequency response function $B(\omega) = \mathscr{B}(e^{-\I\omega})$.
Thus, the moving average process $\eta = \{\eta_t\}$ admits the spectral density operator
$$ \mathscr{F}_\omega^\eta = \frac{1}{2\pi} \mathscr{B}(e^{-\I\omega}) \mathcal{S} \mathscr{B}(e^{-\I\omega})^*$$
by Proposition~\ref{prop:filtration_sum_l^1}.
Obviously, the moving average process $\eta = \{\eta_t\}$ is $q$-correlated, i.e. $\mathscr{R}_h^\eta = 0$ for $|h|>q$, and therefore satisfies the conditions \eqref{eq6:weak_dependence_op_norm} and \eqref{eq6:weak_dependence_traces}, and it is easy to verify that $\mathscr{F}^\eta \in L^{\infty}([0,2\pi], \mathcal{L}_1(\HC) )$.

We now wish to invert \eqref{eq6:FARMA_filtrations} and write the process $X$ as
\begin{equation}\label{eq6:FARMA_filtrations_inverted}
X_t = \mathscr{A}^{-1}(\Delta) \left[ \mathscr{B}(\Delta) \epsilon_t \right] = \mathscr{A}^{-1}(\Delta) \eta_t.
\end{equation}

As part of their existence proof, \citet{klepsch2017prediction}[Theorem 3.8] defined a state space process representation of \eqref{eq6:FARMA_def} as a process in the product space $\HR^p$
$$
\underbrace{\begin{bmatrix}
X_t \\ X_{t-1} \\ \vdots \\ X_{t-p+1} 
\end{bmatrix}}_{\Xi_t}
=
\underbrace{
\begin{bmatrix}
\mathcal{A}_1 & \cdots & \mathcal{A}_{p-1} & \mathcal{A}_p \\
\Id &   &  & 0 \\
  & \ddots &  & \vdots \\
 &  & \Id & 0 \\
\end{bmatrix}
}_{\tilde{\mathcal{A}}}
\underbrace{
\begin{bmatrix}
X_{t-1} \\ X_{t-2} \\ \vdots \\ X_{t-p} 
\end{bmatrix}
}_{\Xi_{t-1}}
+
\underbrace{
\begin{bmatrix}
\eta_t \\ 0 \\ \vdots \\ 0
\end{bmatrix}
}_{\tilde{\eta}_t}
,\qquad t\in\mathbb{Z}.
$$
They showed that the process $\Xi$ can be written as
\begin{equation}\label{eq6:FARMA_state_space_solved}
\Xi_t = \sum_{j=0}^\infty \tilde{\mathcal{A}}^j \tilde{\eta}_{t-j},\qquad t\in\mathbb{Z},
\end{equation}
where
\begin{equation}\label{eq6:FARMA_state_space_solved_summability}
\sum_{j=0}^\infty \| \tilde{\mathcal{A}}^j \|_{\mathcal{L}(\HR^p)} < \infty
\end{equation}
by the assumption \eqref{eq6:FARMA_condition}. Set $P_1$ to be the projection operator onto the first component:
\begin{eqnarray*}
P_1: & \HR^p &\to\quad \HR, \\
& (f_1,\dots,f_n) &\mapsto\quad f_1.
\end{eqnarray*}
Applying $P_1$ to \eqref{eq6:FARMA_state_space_solved} yields
$$ X_t = \sum_{j=0}^\infty P_1 \tilde{\mathcal{A}}^j P_1^* \eta_{t-j} $$
which essentially means that the filter $\mathscr{A}(\Delta)^{-1}$ is given by $( \mathscr{A}(\Delta)^{-1})_s = P_1 \tilde{\mathcal{A}}^s P_1^*$ for $s\geq 0$ and zero otherwise. Moreover, \eqref{eq6:FARMA_state_space_solved_summability} implies
\begin{equation}\label{eq6:FARMA_ar_filter_summability}
\sum_{s\in\mathbb{Z}} \left\| [ \mathscr{A}(\Delta)^{-1}]_s \right\|_{\mathcal{L}(\HR)} < \infty.
\end{equation}
Finally, the application of Proposition \ref{prop:filtration_sum_l^1} onto the filter $\mathscr{A}(\Delta)^{-1}$ and functional time series $\eta$ gives us the spectral density of $X$ given by the formula \eqref{eq6:FARMA_spectral_density_operator}. Moreover, because $\eta$ is $q$-correlated, it trivially satisfies the conditions \eqref{eq6:weak_dependence_op_norm} and \eqref{eq6:weak_dependence_traces} with $\mathscr{R}_h^\eta$, therefore the \FARMA{} process $X$ also satisfies the weak dependence conditions  \eqref{eq6:weak_dependence_op_norm} with $\mathscr{R}_h^X$.
\end{proof}

\subsection{Proof of Theorem~\ref{thm:FARFIMA-part-ii}}
\label{subsec:proof_of_thm:FARFIMA}

\begin{proof}

Building upon the results of Theorem~\ref{thm:FARMA-part-i} we write the \FARMA{} process as
$$ X_t = \mathscr{A}(\Delta)^{-1} \mathscr{B}(\Delta) \epsilon_t = \mathscr{A}(\Delta)^{-1} \eta_t $$
where $ \eta_t = \mathscr{B}(\Delta) \epsilon_t $ is the functional moving average process.
Formally inverting the filter \eqref{eq6:FARIMA_def} yields
$$ \tilde{X}_t = (\Id-\Delta)^{-d} X_t = (\Id-\Delta)^{-d}\mathscr{A}(\Delta)^{-1} \eta_t $$

Following the proof of \citet{hosking1981fractional}[Theorem 1],
define the function $c(z) = (1-z)^{-d},\,z\in\mathbb{C}$. Then the power series expansion of $c$ converges for $|z|\leq 1$  as long as $d<1/2$ and we can write $c(z) = \sum_{k=0}^\infty c_k z^k,\,|z|\leq 1$.
Moreover, using the binomial expansion for $(1-z)^{-d}$ it can be shown \citep{hosking1981fractional} that the coefficients satisfy
\begin{equation}\label{eq6:fi_filter_coefficients_decay}
c_k \sim \frac{ k^{d-1} }{ (d-1)! }, \qquad\text{as}\quad k\to\infty.
\end{equation}

Define with the filter $\mathcal{C} = \{ \mathcal{C}_k \}_{k\in\mathbb{Z}}$ with filter coefficients $\mathcal{C}_k = c_k \Id$ for $k\in\mathbb{N}_0$ where $\Id$ is the identity operator on $\HC$, and zero otherwise.
Obviously $\mathcal{C} = (\Id-\Delta)^{-d}$ in the sense of equality of filters. By the asymptotic relation \eqref{eq6:fi_filter_coefficients_decay}, the filter satisfies 
\begin{equation}\label{eq6:fi_filter_summability}
\sum_{k\in\mathbb{Z}}  \| \mathcal{C}_k \|^2_{\mathcal{L}(\HC)}  < \infty.
\end{equation}
The convolution of the filters $\mathcal{C}$ and $\mathscr{A}(\Delta)^{-1}$, denoted as $\mathcal{D} = \mathcal{C} \ast \mathscr{A}(\Delta)^{-1}$, is given by
$$
\mathcal{D}_s =
\begin{cases} 
      \sum_{k=0}^s \mathcal{C}_k  \left[ \mathscr{A}(\Delta)^{-1} \right]_{s-k}, & s\geq 0,\\
      0, & s<0.
\end{cases}
$$
By way of Young's convolution inequality \citep{hewitt2012abstract}[Theorem 20.18], \eqref{eq6:FARMA_ar_filter_summability} and \eqref{eq6:fi_filter_summability} imply
$$ \sum_{k\in\mathbb{Z}} \| \mathcal{D}_k \|^2_{\mathcal{L}(\HC)} < \infty.$$

Because the moving average process $\eta_t$ is $q$-correlated, we apply Proposition~\ref{prop:filtration_sum_l^2} and obtain the existence and stationary of the \FARFIMA{} process defined by the filter
\begin{align*}
\tilde{X}_t &= \mathcal{D} \eta_t \\
&= \mathcal{C} \left[ \mathscr{A}(\Delta)^{-1} \eta_t \right].
\end{align*}
Moreover, the process $\tilde{X}$ admits the weak spectral density $\mathscr{F}^{\tilde{X}} \in L^1( [0,2\pi], \mathcal{L}_1(\HC) )$ given by
\begin{align*}
\mathscr{F}^{\tilde{X}}_\omega &= \frac{1}{2\pi} \mathcal{D}(\omega) \mathscr{F}_\omega^\eta \mathcal{D}(\omega)^* \\
&= \frac{1}{2\pi} c(e^{-\I \omega}) \mathscr{A}(e^{-\I\omega})\mathscr{B}(e^{-\I\omega}) \mathcal{S}
\mathscr{B}(e^{-\I\omega})^* \left[\mathscr{A}(e^{-\I\omega})^*\right]^{-1}  \overline{ c(e^{-\I \omega}) } \\
&= \frac{1}{2\pi} 
\left[ 2\sin\left(\omega/2\right) \right]^{-2d}
\mathscr{A}(e^{-\I\omega})\mathscr{B}(e^{-\I\omega}) \mathcal{S}
\mathscr{B}(e^{-\I\omega})^* \left[\mathscr{A}(e^{-\I\omega})^*\right]^{-1} ,
\end{align*}
for $\omega\in(0,2\pi)$,
where we have used that $c(e^{-\I\omega}) \Id = (1-e^{-\I\omega})^{-d} \Id = \sum_{k=0}^\infty \mathcal{C}_k e^{-\I k\omega}$ is the frequency response function of the filter $\mathcal{C}$ and
$c(e^{-\I\omega}) \overline{c(e^{-\I\omega})} = |1-e^{-\I\omega}|^{-2d} = [ 2\sin(\omega/2) ]^{-2d}$.
\end{proof}

\subsection{Proof of Theorem~\ref{theorem6:abstract_method}}
\label{subsec:proof_of_thm:abstract_method}
\begin{proof}
The Gaussianity, stationarity, and mean-zero properties of $X_1,\dots,X_T$ are clear thanks to linearity.

First we show that the functional time series defined by \eqref{eq6:simulation_iFFT} is real-valued.
For $k=1,\dots,T/2-1$ we have that
$$ Z_k e^{\I t\omega_k} + Z_{T-k} e^{\I t\omega_{T-k}} =
Z_k e^{\I t\omega_k} + \overline{Z_k} e^{-\I t\omega_k} =
2\Re\{ Z_k e^{\I t\omega_k} \} \quad\in\mathbb{R}. $$
For $k=T/2$ or $k=T$, the spectral density operator $\mathscr{F}^X_{\omega_k}$ is real, thus $Z_k$ is real-valued, and $e^{\I t\omega} \in \{-1,1\}$ for $\omega\in\{\pi,2\pi\}$. Therefore \eqref{eq6:simulation_iFFT} defines a real-valued functional time series.

Let us calculate the lag-$h$ autocovariance operators of \eqref{eq6:simulation_iFFT} for $h\in\mathbb{N}$.
\begin{align}
\nonumber
\Ez{ X_{t+h} \otimes X_t }
&= \frac{\pi}{T} \Ez{ \left(\sum_{k=1}^T Z_k e^{\I(t+h)\omega_k}\right) \otimes \left(\sum_{l=1}^T Z_l e^{\I t\omega_{l}}\right) }\\
\label{eq6:proof_theorem6_eq1}
&= \frac{\pi}{T} \sum_{k=1}^T \sum_{l=1}^T \Ez{ Z_k \otimes Z_l } e^{\I (t+h)\omega_k }e^{-\I t\omega_{l}}
\end{align}

We shall calculate the term $\Ez{ Z_k \otimes Z_l }$ on the right-hand side of \eqref{eq6:proof_theorem6_eq1}.
Firstly, $\Ez{ Z_k \otimes Z_k } = 2\mathscr{F}^X_{\omega_k}$ for $k\in \{T/2, T\}$, and $\Ez{ Z_k \otimes Z_l } = 0$ for $k\in \{T/2, T\}$ and $l\neq k$.

Secondly, fix $k\in\{1,\dots, T/2\}$. Then 
\begin{align*}
\Ez{ Z_k \otimes Z_k } &= 
\Ez{ Z_k' \otimes Z_k' + \I Z_k'' \otimes Z_k' - \I Z_k' \otimes Z_k'' + Z_k'' \otimes Z_k''} = 2 \mathscr{F}^X_{\omega_k}, \\
\Ez{ Z_k \otimes Z_{T-k} } &= 
\Ez{ Z_k' \otimes Z_k' + \I Z_k'' \otimes Z_k' + \I Z_k' \otimes Z_k'' - Z_k'' \otimes Z_k''} = \\
&= \Ez{ Z_k' \otimes Z_k' - Z_k'' \otimes Z_k''} = 0.
\end{align*}
Furthermore, for $l \notin \{k, T-k\}$, we have $\Ez{ Z_k \otimes Z_l } = 0$ from the independence of $Z_k$'s.

We continue with the calculations on \eqref{eq6:proof_theorem6_eq1} as
\begin{equation}\label{eq6:proof_theorem6_eq2}
\Ez{ X_{t+h} \otimes X_t } = \frac{2\pi}{T} \sum_{k=1}^T \mathscr{F}^X_{\omega_k} e^{\I h \omega_k}.
\end{equation}

The right-hand side of \eqref{eq6:proof_theorem6_eq2} constitutes the Riemann sum of the integral \eqref{eq6:spectral_density_operator_inverse_formula}. The convergence of the Riemann sums \eqref{eq6:proof_theorem6_eq2}, as $T\to\infty$, towards \eqref{eq6:spectral_density_operator_inverse_formula} is justified by the assumption \ref{item:theorem6:abstract_method:item_ii} $\mathscr{F}^X_\omega \in L^1([0,2\pi],\mathcal{L}_1(\HC))$. The weak-dependence setting under the assumption \ref{item:theorem6:abstract_method:item_i} is only a special case of the latter but we decided to list them side by side for transparency.
\end{proof}

\section{Functional Filters and Frequency Response Functions}

In this appendix we present the framework of linear filters and their spectral analysis. These technical results are important for derivation of the spectral density operators of the FARMA and FARFIMA processes in Subsections~\ref{subsec:FARMA} and \ref{subsec:FARFIMA}.



Let $X=\{X_t\}_{t\in\mathbb{Z}}$ be a mean-zero stationary functional time series in the separable real Hilbert space $\HR$ with the weak spectral density operator
\begin{equation}\label{eq6:appx_FwX_in_Lp}
\mathscr{F}^X \in L^p([0,2\pi], \mathcal{L}_1(\HC))\qquad\text{for some}\quad p\in(1,\infty].
\end{equation}
Its lag-$h$ autocovariance operators $\mathscr{R}^X_h$ satisfy
$$ \mathscr{R}^X_h = \int_0^{2\pi} \mathscr{F}_\omega^X e^{\I h\omega} \D\omega ,\qquad h\in\mathbb{Z}.$$

A \textit{functional filter}, or simply a \textit{filter}, is a sequence of filter coefficients $\{\theta_s\}_{s\in\mathbb{Z}}$ where $\theta_s \in \mathcal{L}( \HR )$.
Formally, define the filtered functional time series $Y=\{Y_t\}_{t\in\mathbb{Z}}$ as
\begin{equation}\label{eq6:appx_filtered_Y}
Y_t = \sum_{s\in\mathbb{Z}} \theta_s X_{t-s},\qquad t\in\mathbb{Z},
\end{equation}
and the \textit{frequency response function} of $\{\theta_s\}$ as
\begin{equation}\label{eq6:appx_frequency_response}
\Theta(\omega) = \sum_{s\in\mathbb{Z}} \theta_s e^{-\I s\omega},\qquad \omega\in[0,2\pi],
\end{equation}
provided \eqref{eq6:appx_filtered_Y} and \eqref{eq6:appx_frequency_response} converge in an appropriate sense which is justified by the following proposition. 

\begin{proposition}\label{prop:filtration_sum_l^1}
Assume \eqref{eq6:appx_FwX_in_Lp} and that the filter $\{\theta_s\}$ satisfies 
$$ \sum_{s\in\mathbb{Z}} \| \theta_s \|_{\mathcal{L}(\HR)} < \infty.$$

Then the sum on the right-hand side of \eqref{eq6:appx_filtered_Y} converges with respect to $\E\|\cdot\|^2$ and $Y=\{Y_t\}_{t\in\mathbb{Z}}$ is a second-order stationary mean-zero functional time series with values in $\HR$.
Moreover, the sum on the right-hand side of \eqref{eq6:appx_frequency_response} converges in $\mathbb{M}$ (defined in Section~\ref{subsec:CKL})
and the weak spectral density operator $\mathscr{F}^Y \in L^1( [0,2\pi], \mathcal{L}_1(\HC) )$ 
of the functional time series $Y=\{Y_t\}_{t\in\mathbb{Z}}$ is given by
\begin{equation}\label{eq6:appx_filtered_Y_spec_density}
\mathscr{F}^Y_\omega = \Theta(\omega) \mathscr{F}^X_\omega \Theta(\omega)^*, \qquad \omega\in[0,2\pi]
\end{equation}
and the lag-$h$ autocovariance operators of $Y$ are given by
\begin{equation}\label{eq6:appx_filtered_Y_autocovariance}
\mathscr{R}_h^Y = \Ez{ Y_h \otimes Y_0 } = \int_0^{2\pi} \mathscr{F}_\omega^Y e^{\I t\omega}\D\omega, \qquad h\in\mathbb{Z}.
\end{equation}

Furthermore, if \eqref{eq6:weak_dependence_op_norm} and \eqref{eq6:weak_dependence_traces} hold for the time series $X$,
then $\sum_{h\in\mathbb{Z}} \| \mathscr{R}_h^Y \|_{\mathcal{L}(\HR)} < \infty$.
\end{proposition}

\begin{proposition}\label{prop:filtration_sum_l^2}
Assume that the functional time series $X$ is $m$-correlated for some $m\in\mathbb{N}$, i.e. $\mathscr{R}^X_h = 0$ for $|h|>m$, and the filter $\{\theta_s\}$ satisfies 
$$ \sum_{s\in\mathbb{Z}} \| \theta_s \|_{\mathcal{L}(\HR)}^2 < \infty.$$
Then the sum on the right-hand side of \eqref{eq6:appx_filtered_Y} converges with respect to $\E\|\cdot\|^2$ and $Y=\{Y_t\}_{t\in\mathbb{Z}}$ is a second-order stationary mean-zero functional time series with values in $\HR$, the right-hand side of \eqref{eq6:appx_frequency_response} converges in $\mathbb{M}$,
the weak spectral density operator $\mathscr{F}^Y \in L^1( [0,2\pi], \mathcal{L}_1(\HC) )$ is given by \eqref{eq6:appx_filtered_Y_spec_density} and the inverse formula \eqref{eq6:appx_filtered_Y_autocovariance} holds.
\end{proposition}

\begin{proof}[Proof of propositions \ref{prop:filtration_sum_l^1} and \ref{prop:filtration_sum_l^2}]
The stated results are a simplified version of \citet{tavakoli2014}[Theorem 2.5.5, Remark 2.5.6].
\end{proof}

\section{Supplementary Figures for Examples \ref{example6:long_range_FARFIMA} and \ref{example6:FARMA_lowrank}}
\label{sec:supplementary_figures}

Figure~\ref{fig6:FARIMA_trajectories} displays the trajectories of the FARFIMA(1, 0.2, 0) process simualted in Example \ref{example6:long_range_FARFIMA} while Figures \ref{fig6:FARIMA_kernels} and \ref{fig6:FARMA_kernels} depict the kernels of the integral operators used in Examples \ref{example6:long_range_FARFIMA} and \ref{example6:FARMA_lowrank}.
Figures~\ref{fig6:FARFIMA} and \ref{fig6:FARMA} illustrate the results on  simulation accuracy discussed in Examples~\ref{example6:long_range_FARFIMA} and \ref{example6:FARMA_lowrank}.

\begin{figure}[hb]
\centering
\includegraphics[width=1\textwidth]{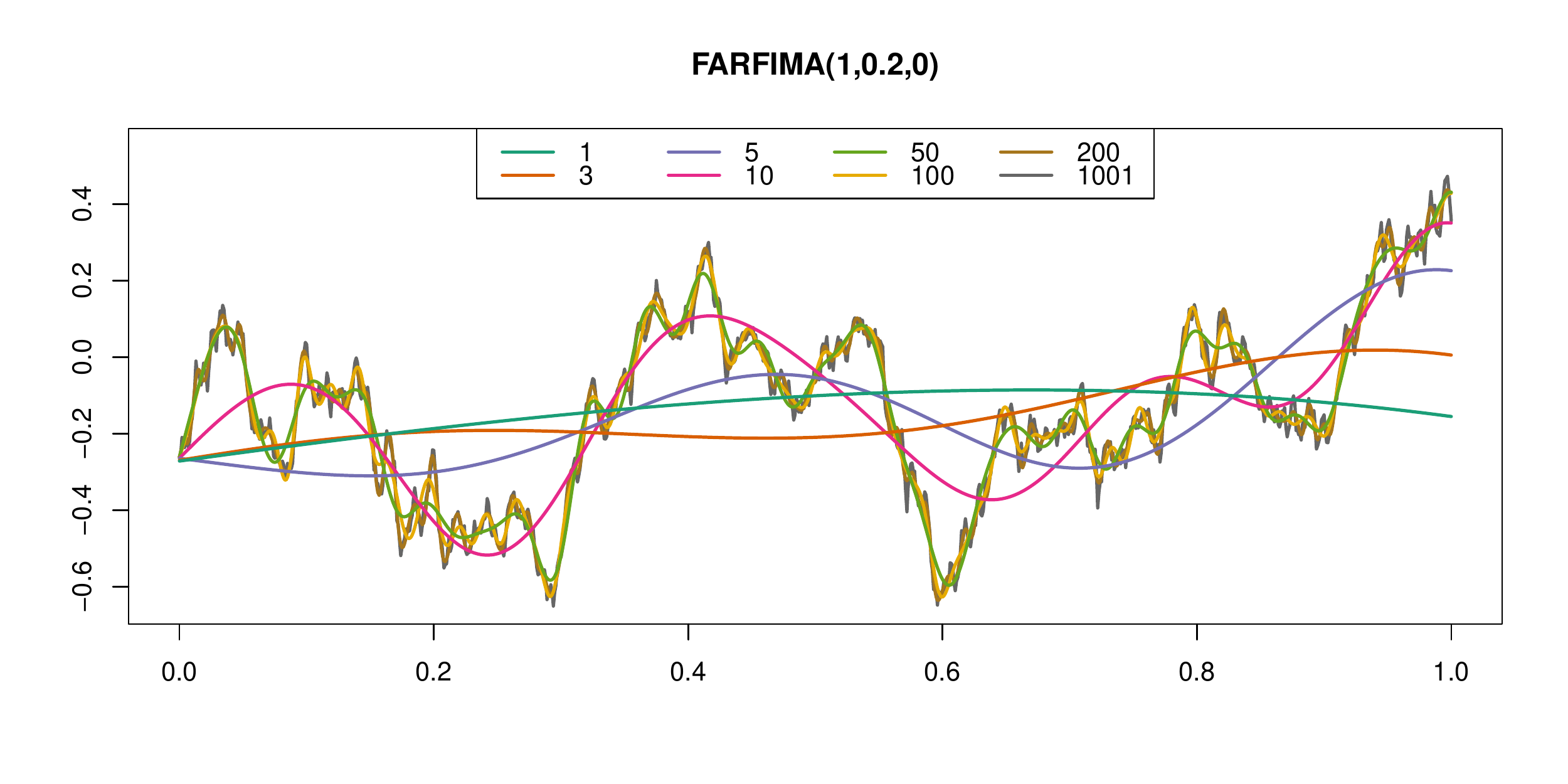}
\caption[Simulated trajectories of the FARFIMA(1,0.2,0) in Example~\ref{example6:long_range_FARFIMA}]{
Sample trajectories $X_1(\cdot)$ of the \textbf{long-range dependent FARFIMA(1,0.2,0)} process defined in Example~\ref{example6:long_range_FARFIMA} with varying number of $N$ chosen in the truncation of \eqref{eq6:simulation_filtered_Y_k}. Simulated with $T=100$ and the grid resolution $M=1001$.
}
\label{fig6:FARIMA_trajectories}
\end{figure}

\begin{figure} 
\centering
\includegraphics[width=0.7\textwidth]{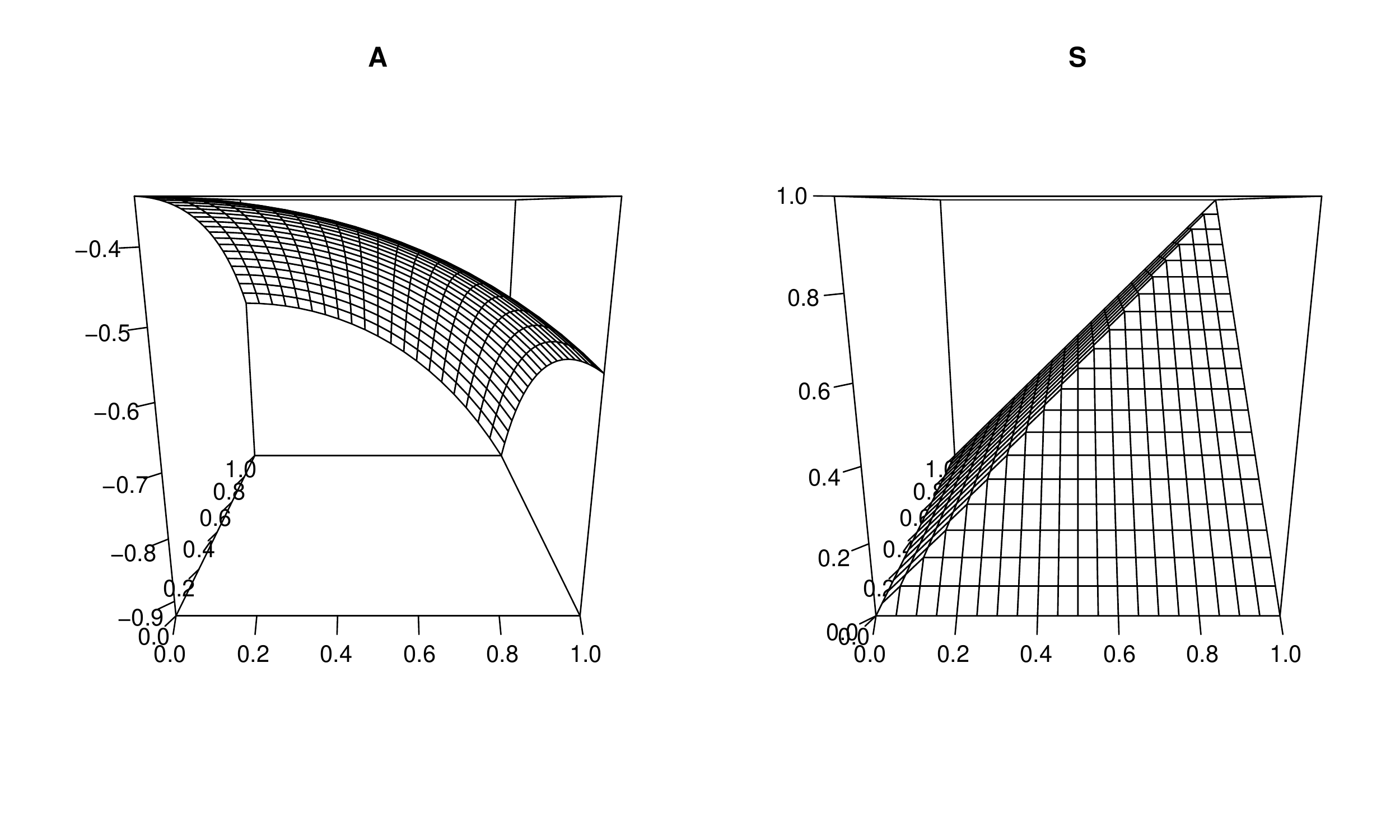}
\caption[The kernels of the functional parameters of the FARFIMA(1,0.2,0) process of Example~\ref{example6:long_range_FARFIMA}]{The kernels of the autoregressive operator and the innovation covariance operator for the FARFIMA(1,0.2,0) process scrutinized in Example~\ref{example6:long_range_FARFIMA}.
}
\label{fig6:FARIMA_kernels}
\end{figure}

\begin{figure}
\centering
\includegraphics[width=1\textwidth]{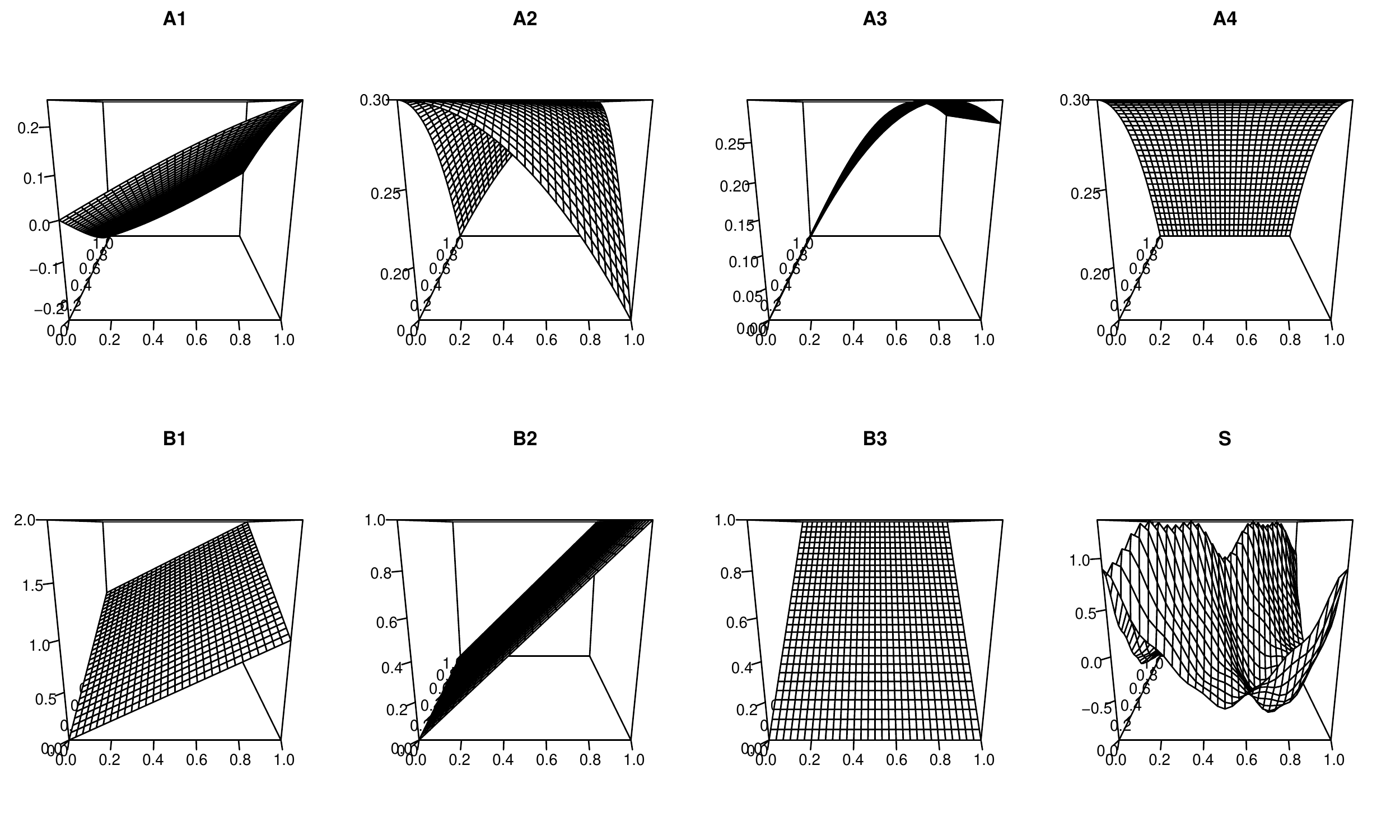}
\caption[The kernels of the functional parameters of the FARMA(4,3) process in Example~\ref{example6:FARMA_lowrank}]{The kernels of the autoregressive operators, moving average operators, and the innovation covariance operator for the FARMA(4,3) process scrutinized in Example~\ref{example6:FARMA_lowrank}.
}
\label{fig6:FARMA_kernels}
\end{figure}

\begin{figure}
\centering
\includegraphics[width=1\textwidth]{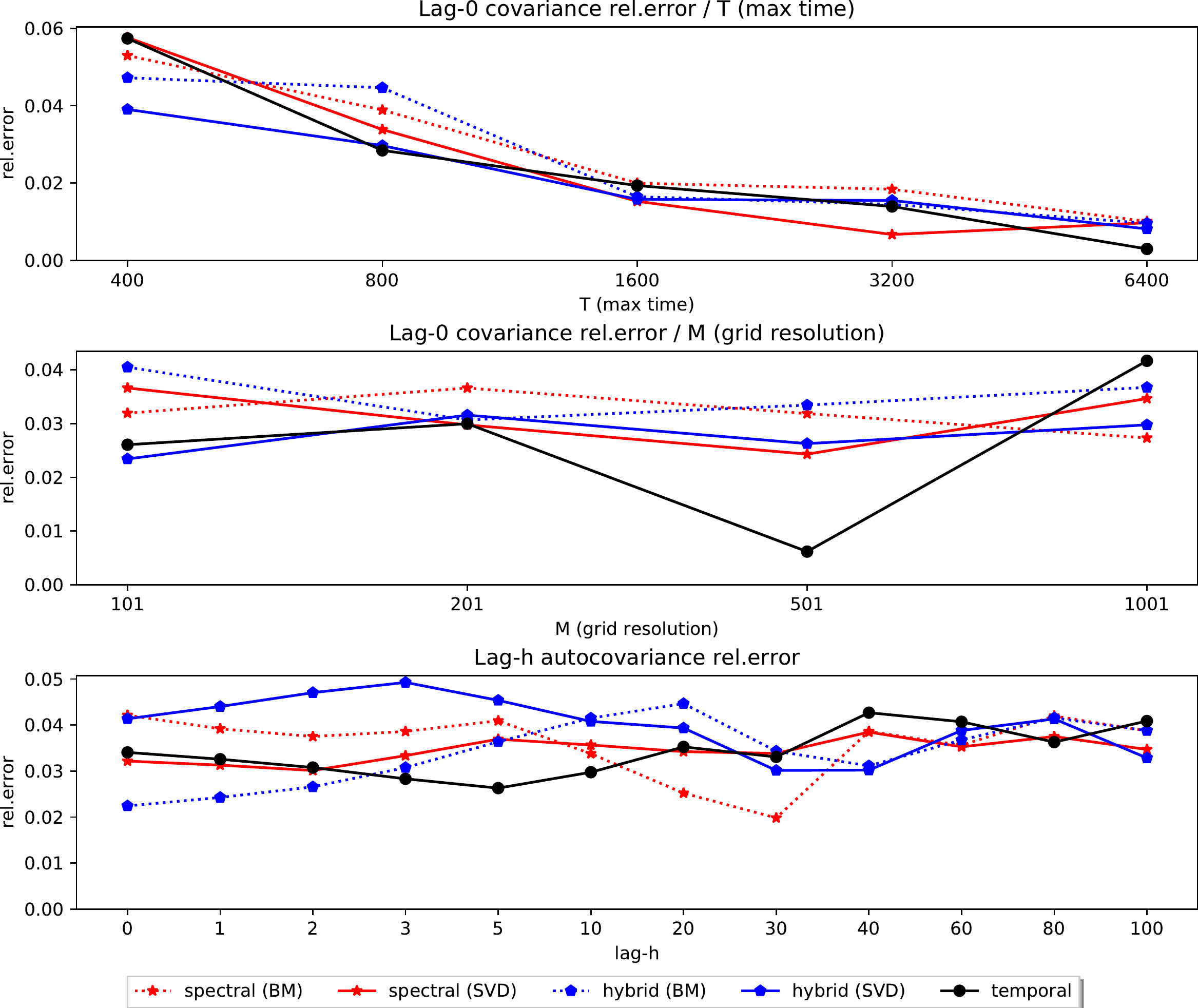}
\caption[The simulation accuracy for the FARFIMA(1,0.2,0) in Example~\ref{example6:long_range_FARFIMA} and its dependence on the simulation parameters]{
The dependence of the \textbf{simulation accuracy} (relative error defined in \eqref{eq6:RMSE_simulation}) for the long-range dependent FARFIMA(1,0.2,0) process defined in Example~\ref{example6:long_range_FARFIMA} on simulation parameters
\textbf{Top:} The dependence of the lag-$0$ covariance operator simulation accuracy on time horizon $T\in\{400,800,1600,3200,6400\}$, while the spatial resolution is set $M=101$.
\textbf{Center:} the dependence of the lag-$0$ covariance operator simulation accuracy on the grid size $M\in\{101,201,501,1001\}$, while the the time horizon is set $T=800$.
\textbf{Bottom:} the dependence of the lag-$h$ covariance operator simulation accuracy on $h\in\{0,1,2,3,5,10,20,30,40,60,80,100\}$, with $T=800$ and $M=101$.
}
\label{fig6:FARFIMA}
\end{figure}

\begin{figure}
\centering
\includegraphics[width=1\textwidth]{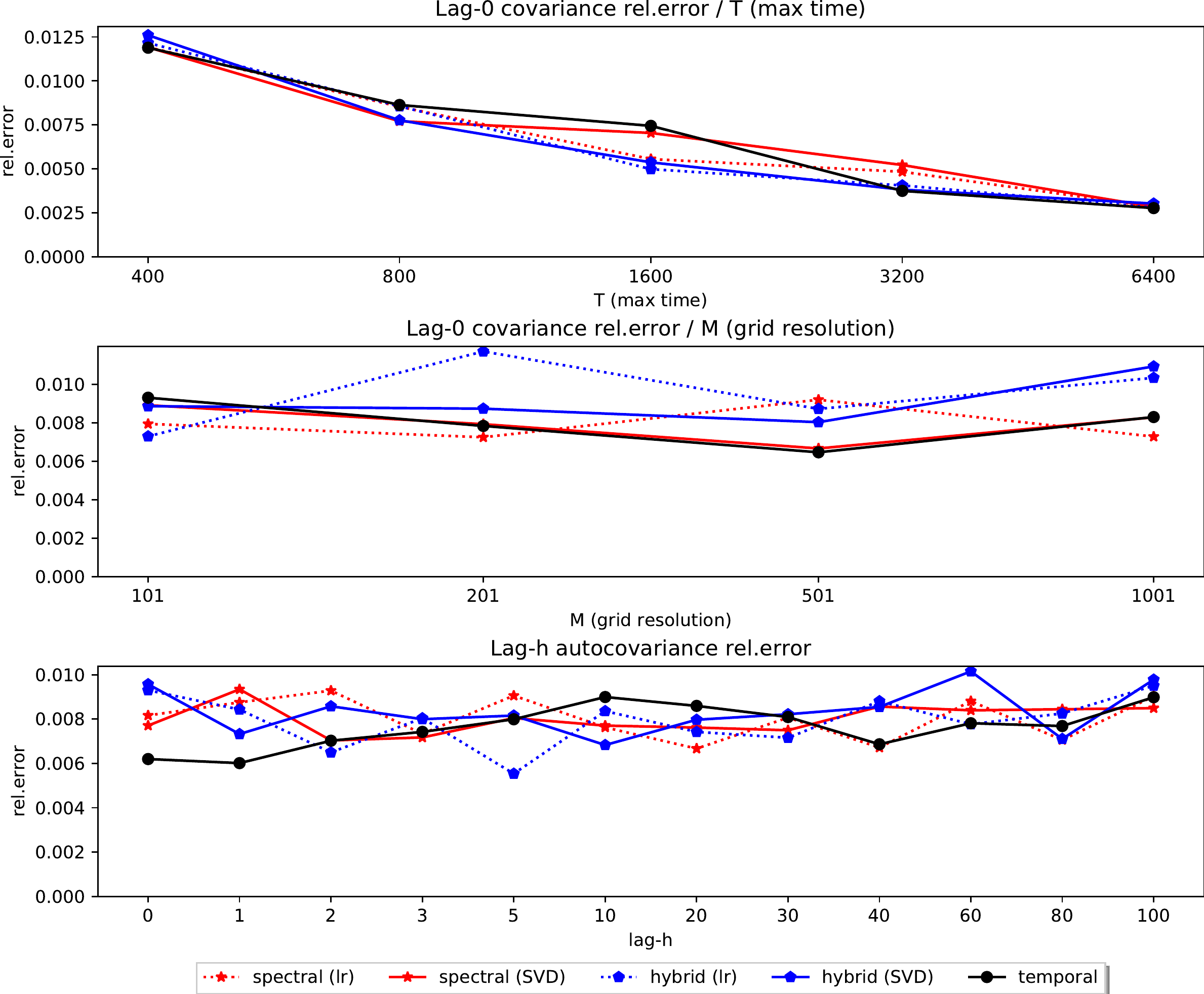}
\caption[The simulation accuracy for the FARMA(4,3) in Example~\ref{example6:FARMA_lowrank} and its dependence on the simulation parameters]{
The dependence of the \textbf{simulation accuracy} (relative error defined in \eqref{eq6:RMSE_simulation}) for the FARMA(4,3) process defined in Example~\ref{example6:FARMA_lowrank} on simulation parameters
\textbf{Top:} The dependence of the lag-$0$ covariance operator simulation accuracy on time horizon $T\in\{400,800,1600,3200,6400\}$, while the spatial resolution is set $M=101$.
\textbf{Center:} the dependence of the lag-$0$ covariance operator simulation accuracy on the grid size $M\in\{101,201,501,1001\}$, while the the time horizon is set $T=800$.
\textbf{Bottom:} the dependence of the lag-$h$ covariance operator simulation accuracy on $h\in\{0,1,2,3,5,10,20,30,40,60,80,100\}$, with $T=800$ and $M=101$.
}
\label{fig6:FARMA}
\end{figure}


\FloatBarrier

\bibliographystyle{plainnat}
\bibliography{biblio}

\end{document}